\documentclass[twoside,11pt]{article}
\usepackage{multirow}
\usepackage{bm}
\usepackage{amsmath, amssymb, amsfonts,amsthm}
\usepackage{array}
\usepackage{epsfig}
\usepackage{enumerate}
\usepackage{enumitem}

\usepackage{bbm}
\usepackage{color}
\usepackage{float}
\usepackage[mathscr]{euscript}
\usepackage{mathtools}
\usepackage{graphicx,subcaption}
\usepackage{physics}
\usepackage{framed}
\usepackage{setspace}
\usepackage{mdframed}

\usepackage[numbers]{natbib} 

\usepackage{thmtools}
\usepackage{environ}
\usepackage{mathptmx}
\usepackage[draft]{todonotes} 
  \usepackage{times}
\usepackage{pgfplots}
\usetikzlibrary{matrix}
\usepackage[toc,page]{appendix}

  \usetikzlibrary{positioning,calc,fit,shapes.geometric,patterns}
  \pgfdeclarelayer{background}
  \pgfdeclarelayer{foreground}
  \pgfsetlayers{background,main,foreground}
  \tikzstyle{vec}=[circle,inner sep=1pt,outer sep=-1pt,fill]
  \tikzstyle{border}=[thick]
  \tikzstyle{favborder}=[border,dotted]
  \tikzstyle{exclborder}=[border,dashed]
\usepackage{etex,etoolbox}
\usepackage{hyperref}
 \usepackage{calrsfs}

  \usetikzlibrary{positioning,calc,fit,shapes.geometric,patterns}
  \pgfdeclarelayer{background}
  \pgfdeclarelayer{foreground}
  \pgfsetlayers{background,main,foreground}
  \tikzstyle{vec}=[circle,inner sep=1pt,outer sep=-1pt,fill]
  \tikzstyle{border}=[thick]
  \tikzstyle{favborder}=[border,dotted]
  \tikzstyle{exclborder}=[border,dashed]
\usepackage{etex,etoolbox}
\usepackage{hyperref}

\hypersetup{colorlinks=true,linkcolor=blue,citecolor=magenta}

\usepackage{etex,etoolbox}
\usepackage{hyperref}

\newenvironment{proofof}[1]{\begin{trivlist}\item[\hskip\labelsep{\textsc{Proof
  of {#1}.\ }}]}{\hspace*{\fill} {\qed}\end{trivlist}}

\newenvironment{tbs}{%
   \small\tt
   \begin{enumerate}[$\blacktriangleright$]}{\end{enumerate}}
\newcommand{\btbs}{\begin{tbs}}                                                                      
\newcommand{\etbs}{\end{tbs}}

\newcommand{\Reals}{\mathbb{R}} 

\newcommand{\complex}{\mathbb{C}}

 \DeclareMathOperator{\posi}{posi}

  \DeclareMathOperator{\cl}{cl}
\newcommand{\pspace}{\Omega}

\newcommand{\reals}{\Reals}

\newcommand{\gambles}{\mathcal{L}}

\newcommand{\He}{\mathcal{H}}

\newcommand{\theory}{\mathcal{T}}
\newcommand{\btheory}{\mathcal{T}^\star}
\newcommand{\domain}{\mathcal{K}}
\newcommand{\assess}{\mathcal{G}}
\newcommand{\bdomain}{\mathcal{C}}

\newcommand{\nonnegative}{\gambles^{\geq} }
\newcommand{\bnonnegative}{\Sigma^{\geq}}

\newcommand{\negative}{\gambles^{<} }
\newcommand{\bnegative}{\Sigma^{<} }
\newcommand{\states}{\mathscr{S}}

\newtheorem{lemma}{Lemma}

\newtheorem{theorem}{Theorem}
\newtheorem{definition}{Definition}
\newtheorem{proposition}{Proposition}
\newtheorem{example}{Example}

\usepackage{tikz-cd}
\title{Computational Complexity and the Nature of Quantum Mechanics \\ (Extended version)}
\author{A. Benavoli, A. Facchini, M. Zaffalon\\
IDSIA, Manno, Switzerland}

\begin{document}

\maketitle

\begin{abstract}
Quantum theory (QT) has been confirmed by numerous experiments, yet we still cannot fully grasp the meaning of the theory. As a consequence, the quantum world appears  to us paradoxical.
Here we shed new light on QT by having it follow from two main postulates
(i) the theory should be logically consistent; (ii) inferences in the theory should be computable in polynomial time.
The first postulate is what we  require to each well-founded mathematical theory. The computation postulate defines the physical component of the theory. 
We show that the computation postulate is the only true divide between QT, seen as a generalised theory of probability,
and classical probability. 
All quantum paradoxes, and entanglement
in particular, arise from the clash of trying to reconcile a computationally intractable, somewhat idealised,  
theory (classical physics) with a computationally tractable theory (QT) or, in other
words, from regarding physics as fundamental rather than computation.


\end{abstract}
 
\tableofcontents



\section{Introduction}


Quantum theory (QT) is one of the most fundamental, and accurate, mathematical descriptions of our physical world. It dates back to the 1920s, and in spite of nearly one century passed by since its inception, we do not have a clear understanding of such a theory yet. In particular, we cannot fully grasp the \emph{meaning} of the theory: why it is the way it is. As a consequence, we cannot come to terms with the many paradoxes it appears to lead to --- its so-called ``quantum weirdness''.

This paper aims at finally explaining QT while giving a unified reason for its many paradoxes. We pursue this goal by having QT follow from two main postulates:
\begin{description}
\item[(Coherence)] The theory should be logically consistent. 
\item[(Computation)] Inferences in the theory should be computable in polynomial time. 
\end{description}

The first postulate is what we essentially require to each well-founded mathematical theory, be it physical or not: it has to be based on a few axioms and rules from which we can unambiguously derive its mathematical truths. 
The second postulate will turn out to be central. It requires that there should be an efficient way to execute the theory in a computer.

QT is an abstract theory that can be studied detached from its physical applications. For this reason, people often wonder which part of QT actually pertains to physics. In our representation, the answer to this question shows itself naturally: the computation postulate defines the physical component of the theory. But it is actually stronger than that: it states that computation is more primitive than physics. 

Let us recall that QT is widely regarded as a ``generalised'' theory of probability. In this paper we make the adjective ``generalised'' precise. In fact, our coherence postulate leads to a theory of  probability, in the sense that it disallows ``Dutch books'': this means, in gambling terms, that a bettor on a quantum experiment cannot be made a sure loser by exploiting inconsistencies in their probabilistic assessments. But probabilistic  inference is in general NP-hard. By imposing the additional postulate of computation, the  theory becomes one of ``computational rationality'': one that is consistent (or coherent), up to the degree that polynomial computation allows. This weaker, and hence more general, theory of probability is QT.

As a result, for a subject living inside QT, all is coherent. For us, living in the classical, and somewhat idealised, probabilistic world (not restricted by the computation postulate), QT displays some inconsistencies: precisely those that cannot be fixed in polynomial time. All quantum paradoxes, and entanglement in particular, arise from the clash of these two world views: i.e., from trying to reconcile an unrestricted theory (i.e., classical physics) with a  theory of computational rationality (quantum theory). Or, in other words, from regarding physics as fundamental rather than computation.

But there is more to it. We show that the theory is ``generalised'' also in another direction, as QT turns out to be a theory of ``imprecise'' probability: in fact, requiring the computation postulate is similar to defining a probabilistic model using only a finite number of moments; and therefore, implicitly, to defining the model as the set of all probabilities compatible with the given moments. In QT, some of these compatible probabilities can actually be \emph{signed}, that is, they allow for ``negative probabilities''. In our setting, these have no meaning per se, they are just a mathematical consequence of polynomially bounded coherence (or rationality).

\subsection{Relations with the literature}
Since its foundation, there have been two main ways to explain the differences between QT and classical probability. The first one, that goes back to Birkhoff and von Neumann \cite{birkhoff1936logic}, explains this differences  with the premise that,
in QT, the Boolean algebra of events is taken over by the ``quantum logic'' of projection operators on a Hilbert space.
The second one is based on the view that the quantum-classical clash is due to the appearance of negative probabilities \cite{feynman1987negative}.

Recently, there has been a research effort, the so-called ``quantum reconstruction'', which amounts to trying to rebuild the theory from more primitive postulates.
The search for alternative axiomatisations of QT has been approached following different avenues: extending
Boolean  logic \cite{birkhoff1936logic,mackey2013mathematical,jauch1963can}, using operational primitives \cite{hardy2011foliable,hardy2001quantum,barrett2007information,chiribella2010probabilistic},
using  information-theoretic postulates \cite{barrett2007information,barnum2011information,van2005implausible,pawlowski2009information,dakic2009quantum,fuchs2002quantum,brassard2005information,mueller2016information},
building upon the subjective foundation of probability \cite{Caves02,Appleby05a,Appleby05b,Timpson08,longPaper, 
Fuchs&SchackII,mermin2014physics,pitowsky2003betting,Pitowsky2006,benavoli2016quantum,benavoli2017gleason} and starting from the phenomenon of quantum nonlocality \cite{barrett2007information,van2005implausible,pawlowski2009information,popescu1998causality,navascues2010glance}. 

A common trait of all these approaches is that of regarding QT as a generalised theory of probability. But why is probability generalised in such a way, and what does it mean? Our paper appears to be the first to show that the answer to this question rests in the computational intractability of classical probability theory contrasted to the polynomial-time complexity of QT.

Note that there have been previous investigations into the computational nature of QT but they have mostly focused on topics of undecidability and of potential computational advantages of non-standard theories involving modifications of quantum theory
\cite{bacon2004quantum,aaronson2004quantum,aaronson2005quantum,chiribella2013quantum}.\footnote{The undecidability results in QT are usually obtained via a limiting argument, as the number of particles goes to infinity (see, e.g., \cite{cubitt2015undecidability}). These results do not apply to our setting as we rather take the stance that the Universe is a finite physical system.} 

\subsection{Outline of the paper}
Section~\ref{sec:dg} is concerned with the coherence principle. We recall how Bayesian probability can be derived (via mathematical duality) from a set of logical axioms. Addressing self-consistency (coherence or rationality) in such a setting is a standard task in logic; in practice, it reduces to prove that a certain real-valued bounded function is non-negative. 

Section~\ref{sec:comp} details the computation principle. We consider the problem of verifying the non-negativity of a function as above. This problem is generally undecidable or, when decidable, NP-hard. We make the problem polynomially solvable by redefining the meaning of (non-)negativity. We give our fundamental theorem (Theorem~\ref{th:fundamental}) showing that the redefinition is at the heart of the clash between Bayesian probability and computational rationality.

We show in Section~\ref{sec:coheqm} that QT is a special instance of computational rationality and hence that Theorem~\ref{th:fundamental} is not only the sole difference between quantum and classical probability, but also  the distinctive reason for all  quantum paradoxes; this latter part is discussed in Section~\ref{sec:entan0}. 
In particular, to give further insight about the quantum-classical clash, in Section~\ref{sec:local} we reconsider the question of local realism in the light of computational rationality; in Section~\ref{sec:witness} we show that the witness function, in the fundamental ``entanglement witness theorem'',  is nothing else than a negative function whose negativity cannot be assessed in polynomial time---whence it is not ``negative'' in QT.

Moreover, using Theorem \ref{th:fundamental}, in Section~\ref{sec:ent_not_only} we devise an example of a computationally tractable theory of probability that is unrelated to QT but that admits entangled states. This shows in addition that the ``quantum logic''  and the ``quasi-probability'' foundations of QT are two faces of the same coin, being natural consequences of the computation principle.

We finally discuss the results in Section~\ref{sec:discussions}. The technical proofs of the paper are in Appendix.

 \section{Desirability}
 \label{sec:dg}
 
\subsection{Coherence postulate}
De Finetti's subjective foundation of probability \cite{finetti1937}   is  based on the notion   of rationality (self-consistency or coherence). This approach has then been further developed in \citep{williams1975,walley1991}, giving rise to the so-called \emph{theory of desirable gambles} (TDG).\footnote{Contrarily to what it may seem, TDG is not an ``exotic'' theory of probability; loosely speaking, it is just an equivalent reformulation of the well-known Bayesian decision theory (\`a la Anscombe-Aumann \cite{anscombe1963}) once this is extended to deal with incomplete preferences \cite{zaffalon2017a,zaffalon2018a}.} In this setting probability is a derived notion in the sense that it can be inferred via mathematical duality from a set of logical axioms that one can interpret 
as rationality requirement in the way a subject, let us call her Alice, accepts gambles on the results of an uncertain experiment. It goes as follow. 

Let $\pspace$ denote the possibility space of an experiment (e.g., $\{Head, Tail\}$ 
or $\complex^n$ in QT). A gamble $g$ on $\pspace$ is a bounded real-valued function of $\pspace$, interpreted as
an uncertain reward. It plays the traditional role of variables or, using a physical parlance, of \emph{observables}. 
In the context we are considering, accepting a gamble $g$ by an agent is regarded as a
commitment to receive, or pay (depending on the sign), $g(\pspace)$ \emph{utiles}\footnote{Abstract units of utility, indicating the satisfaction derived from an economic transaction; we can approximately identify it with money provided we deal with small amounts of it \cite[Sec.~3.2.5]{finetti1974}.} whenever $\pspace$ occurs.
Given this view, if by $\gambles$ we denote the set of all the gambles on $\pspace$, the subset of all non-negative gambles, that is, of gambles for which Alice is never expected to lose utiles, is given by  
$\nonnegative\coloneqq \{g \in \gambles: \inf g\geq0 \}$. Analogously, 
negative gambles, those gambles for which Alice will certainly lose some utiles, even an epsilon, is defined as $\negative\coloneqq \{g \in \gambles: \sup g < 0 \}$.
In what follows, with $\mathcal{G} \subset \gambles$ we denote a finite\footnote{We will comment on the case when $\mathcal{G}$ may not be finite.} set of gambles that Alice finds desirable: these are the gambles that she is willing to accept and thus commits herself to the corresponding transactions.

The crucial question is now to provide a criterion for a set $\mathcal{G}$ of gambles representing assessments of desirability to be called \emph{rational}. Intuitively Alice is rational if she avoids sure losses: that is, if, by considering the implications of what she finds desirable, she is not forced to find desirable a negative gamble. This postulate of rationality is called ``no arbitrage'' in economics and ``no Dutch book'' in the subjective foundation of probability.
In  TDG we formulate it thorough the  notion of logical coherence which, despite the informal interpretation given above, is a purely syntactical (structural) notion. To show this, we need to define an appropriate logical calculus, that is, the tautologies and the inference rules (characterising the set of gambles that Alice must find desirable as a consequence of having desired $\mathcal{G}$ in the first place), and based on it to characterise the family of consistent sets of assessments..

Given that non-negative gambles
may increase Alice's utility without ever decreasing it, we have that:
\begin{enumerate}[label=\upshape A0.,ref=\upshape A0]
\item\label{eq:taut} $\nonnegative$ should always be desirable.
\end{enumerate}
This defines the tautologies of the calculus. We thus  characterise the
set of gambles that we must find desirable as a consequence of having desired $\mathcal{G}$ in
the first place, that is its the deductive closure of a set $\assess$. Those
gambles are the conical hull of gambles in $\mathcal{G}$. Indeed, whenever $f,g$ are desirable for Alice, then any positive linear combination of them should also be desirable (this amounts to assuming that Alice has a linear utility scale, which is a standard assumption in probability):
\begin{equation}
\label{eq:posi}
\posi(\mathcal{G}):=\left\{h \in \gambles: h=\sum_{i=1}^{\ell} \lambda_i g_i ~~\textit{ with }~~ \lambda_i\geq0, ~g_i\in \mathcal{G}, \ell \in \mathbb{N}\right\}.
\end{equation}
Moreover, we can  assume that if Alice find all gambles  of type $g + \epsilon$ desirable, for any arbitrary small positive $\epsilon$, then she should also find $g$ desirable.
This means that the actual deductive closure we are
after is given by the map $N$  associating to $\mathcal{G}$ the set:

\begin{enumerate}[label=\upshape A1.,ref=\upshape A1]
\item\label{eq:NE} $N(\assess)\coloneqq\cl \posi(\nonnegative\cup \mathcal{G})$.
\end{enumerate}
where $\cl$ is the topological closure operator given the supremum norm topology on $\gambles$. The set $N(\assess)$  is the smallest \textbf{closed convex cone} that includes $\nonnegative\cup \mathcal{G}$, and it is called the \emph{natural extension} of $\assess$, and sometimes is simply denoted by $\domain$.
Note that $\domain = \posi(\nonnegative\cup \mathcal{G})$ whenever $\mathcal{G}$ is finite.

In a betting system, a sure loss for an agent is represented by a negative gamble. Indeed, whenever the outcome of the experiment may be, accepting $g \in \negative$ means to accept to pay some non zero utiles. We therefore say that:

\begin{definition}[Coherence postulate]
\label{def:avs}
 A set $\domain$ of desirable gambles is \emph{coherent} if and only if
 \begin{enumerate}[label=\upshape A2.,ref=\upshape A2]
\item\label{eq:sl} $ \negative \cap \domain=\emptyset$.
\end{enumerate}
\end{definition}
\noindent As simple as it looks, expression \ref{eq:sl} alone captures the coherence postulate as formulate in the introduction
 in case of classical probability theory. This will be make precise in Section \ref{subsec:dual}.

The following result, in addition to providing a necessary and sufficient condition for coherence,  states that $-\mathbbm{1}$ can be regarded as playing the role of the Falsum and \ref{eq:sl} can be reformulated as $-\mathbbm{1} \notin \domain$. 
Note that. we have introduced the symbol $\mathbbm{1}$ to distinguish the unitary function in $\gambles$, i.e.,  $\mathbbm{1}(\omega)=1$ for all $\omega \in  \Omega$,
from the scalar (real number) $1$. This will be convenient later in Section \ref{sec:comp}.
It is an immediate consequence of  Theorem 3.8.5 and Claim 3.7.4 in \cite{Walley91}.
\begin{proposition}\label{prop:coco1}
Let $\assess$ be a set of gambles. The following claims are equivalent
\begin{enumerate}
\item $N(\assess)$ is coherent, 
\item $\posi (\assess) \cap \negative = \emptyset$,
\item $-\mathbbm{1} \notin \posi (\assess \cup \nonnegative) $,
\item $g \notin N(\assess)$, for some gamble $g$.
\end{enumerate}
\end{proposition}

Postulate \ref{eq:sl}, which presupposes postulates \ref{eq:taut} and \ref{eq:NE},  provides the normative definition of TDG, referred to by $\theory$. Based on it, in Subsection \ref{subsec:dual} we derive the axioms of classical, Bayesian, probability theory. 
This is simply based on the fact that, geometrically,  $\domain$ is a closed convex cone.  It is thence clear from the above definition that $\nonnegative$ is the minimal coherent set of desirable gambles. It characterises a state of full ignorance -- a subject
without knowledge about $\omega$ should only  accept nonnegative gambles.
Conversely, a coherent set of desirable gambles is called \emph{maximal} if there is no other coherent set of desirable gambles including it. 
In terms of rationality, a maximal coherent set of desirable gambles is a set of gambles that Alice cannot extend  by accepting
other gambles while keeping at the same time rationality. It also represents a situation in
which Alice is sure about the state of the system, as we will show in the next examples and section.

\begin{example}
Let us consider  the toss of a fair coin  $\Omega=\{Head,Tail\}$.
A gamble $g$ in this case has two components $g(Head)=g_1$ and $g(Tail)=g_2$.
 If Alice accepts $g$ then  she commits herself to receive/pay $g_1$ if  the outcome is Heads and   $g_2$ if Tails.
Since a gamble is in this case an element of $\mathbb{R}^2$,  $g=(g_1,g_2)$, we can plot the gambles  Alice accepts  in a 2D coordinate system with coordinate  $g_1$ and $g_2$, see Figure \ref{fig:coin0}.
\ref{eq:taut} says that Alice is willing to accept any gamble $g=(g_1,g_2)$ that, no matter the result of the experiment, may increase her wealth without ever decreasing it, that is with $g_i \geq 0$ -- Alice always accepts the first quadrant, Figure~\ref{fig:coin0}(a). 
Similarly.  Alice does not accept  any gamble $g=(g_1,g_2)$ that will surely decrease her wealth, that is with  $g_i<0$ (this follows by \ref{eq:sl}). In other words, Alice always does not accept  the interior of the third quadrant, Figure~\ref{fig:coin0}(b).   Then we ask Alice about $g=(-0.1,1)$ -- she loses $0.1$ if Heads and wins $1$  if Tails.
Since Alice knows that the coin is fair, she accepts this gamble as well as all the gambles of the form  $\lambda g$ with  $\lambda>0$, because this is just a ``change of currency'' (scaling).  Similarly, she accepts all the gambles $ g + h$ for any  $h \in \gambles^+$, since these gambles are even more favourable for her (additivity). 
Scaling and additivity follow by \ref{eq:NE}.

Now, we can ask Alice about $g=(1,-0.1)$ and the argument is symmetric to the above case.
We therefore obtain the following set of desirable gambles (see Figure~\ref{fig:coin0}(c)):
$\domain_2=\{g \in \mathbb{R}^2 \mid 10g_1+g_2\geq 0 \text{ and } g_1+10g_2\geq 0\}$.
Finally, we can ask Alice about $g=(-1,1)$ -- she loses $1$ if Heads and wins $1$  if Tails.
Since the coin is fair, Alice  accepts  this gamble.  A similar conclusion can be derived for the symmetric gamble $g=(1,-1)$. Figure~\ref{fig:coin0}(d) is her final set of desirable gambles about the experiment concerned with the toss of a fair coin, which in a formula becomes
$\domain_3=\{g \in \mathbb{R}^2 \mid g_1+g_2\geq 0\}$. The resulting closed convex cone is maximal.
Alice does not accept any other gamble. In fact, if Alice  would also accept for instance $h=(-2,0.5)$ then, 
since she has also accepted $g=(1.5,-1)$, i.e., $g\in \domain_3$, she must also accept $g+h$ (because of \ref{eq:NE}).
However, $g+h=(-0.5,-0.5)$ is always negative, Alice always loses utiles in this case. In other words, by accepting $h=(-2,0.5)$ Alice incurs a sure loss -- she is irrational (\ref{eq:sl} does not hold).
\end{example}
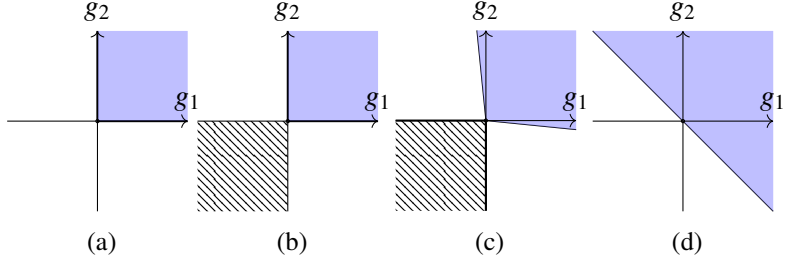
\begin{figure*}[t!]
    \centering
    \begin{subfigure}[t]{0.2\textwidth}
        \centering
\begin{tikzpicture}[scale=1.2]
    \draw[->] (-1,0) coordinate (xl) -- (1,0) coordinate (xu) node[above] {$g_1$};
    \draw[->] (0,-1) coordinate (yl) -- (0,1) coordinate (yu) node[above] {$g_2$};
        \draw (0,0) circle (0.5pt);
    \begin{pgfonlayer}{background}
      \draw[border] (0,0) -- (0,1) coordinate (a1away);
      \draw[border] (0,0) -- (1,0) coordinate (a2away);
      \fill[blue!50,  opacity=0.5] (0,0) -- (1,0) -| (1,1) --(0,1);
    \end{pgfonlayer}
  \end{tikzpicture}
        \caption{}
    \end{subfigure}%
    \begin{subfigure}[t]{0.2\textwidth}
        \centering
\begin{tikzpicture}[scale=1.2]
    \draw[->] (-1,0) coordinate (xl) -- (1,0) coordinate (xu) node[above] {$g_1$};
    \draw[->] (0,-1) coordinate (yl) -- (0,1) coordinate (yu) node[above] {$g_2$};
    \draw (0,0) circle (0.5pt);
        \begin{pgfonlayer}{background}
      \draw[border] (0,0) -- (0,1) coordinate (a1away);
      \draw[border] (0,0) -- (1,0) coordinate (a2away);
      \fill[blue!50,  opacity=0.5] (0,0) -- (1,0) -| (1,1) --(0,1);
    \end{pgfonlayer}
    \begin{pgfonlayer}{background}
      \fill[pattern=north west lines, pattern color=black] (0,0) -- (-1,0) -| (-1,-1) --(0,-1);
    \end{pgfonlayer}
  \end{tikzpicture}
        \caption{}
    \end{subfigure}
        \begin{subfigure}[t]{0.2\textwidth}
        \centering
 \begin{tikzpicture}[scale=1.2]
    \draw[->] (-1,0) coordinate (xl) -- (1,0) coordinate (xu) node[above] {$g_1$};
    \draw[->] (0,-1) coordinate (yl) -- (0,1) coordinate (yu) node[above] {$g_2$};
    \draw (0,0) circle (0.5pt);
    \begin{pgfonlayer}{background}
      \draw[border] (0,0) -- (0,-1) coordinate (a1away);
      \draw[border] (0,0) -- (-1,0) coordinate (a2away);
      \fill[pattern=north west lines, pattern color=black] (0,0) -- (-1,0) -| (-1,-1) --(0,-1);
    \end{pgfonlayer}
    \begin{pgfonlayer}{background}
      \draw[-] (0,0) -- (-0.1,1) coordinate (a1away);
      \draw[-] (0,0) -- (1,-0.1) coordinate (a2away);
      \fill[blue!50,  opacity=0.5] (0,0) -- (a1away) -| (a2away) --(0,0);
    \end{pgfonlayer}
  \end{tikzpicture}
          \caption{}
    \end{subfigure}
        \begin{subfigure}[t]{0.2\textwidth}
        \centering
  \begin{tikzpicture}[scale=1.2]
    \draw[->] (-1,0) coordinate (xl) -- (1,0) coordinate (xu) node[above] {$g_1$};
    \draw[->] (0,-1) coordinate (yl) -- (0,1) coordinate (yu) node[above] {$g_2$};
    \draw (0,0) circle (0.5pt);
    \begin{pgfonlayer}{background}
      \draw[-] (0,0) -- (-1,1) coordinate (a1away);
      \draw[-] (0,0) -- (1,-1) coordinate (a2away);
      \fill[blue!50,  opacity=0.5] (0,0) -- (a1away) -| (a2away) --(0,0);
    \end{pgfonlayer}
  \end{tikzpicture}
        \caption{}
    \end{subfigure}
    \caption{Alice's sets of coherent set of desirable gambles for the experiment of tossing a fair coin.}
    \label{fig:coin0}
\end{figure*}



\subsection{Inference}
\label{sec:primalinference}
In the operational interpretation of $\theory$, agents can buy/sell gambles from/to each other.
Therefore, an agent must be able to determine the selling/buying prices for gambles.
This can be formulated as an inference procedure on  $\domain$. For simplicity, we consider finite sets of assessments, and denote by $| A|$  the cardinality of a finite set $A$.

\begin{definition}
 Let $\assess$ be a finite set of assessments of desirability, and $N(\assess)$ be a coherent set of
desirable gambles. Given  $f \in \gambles$, we denote with
\begin{equation}
\label{eq:lp}
\begin{aligned}
 \underline{E}(f):=&\sup_{\gamma_0\in \reals, \lambda_i \in \reals^+} \gamma_0\\
 &s.t:\\
 &f(\omega) - \gamma_0  -\sum\limits_{i=1}^{|\mathcal{G}|} \lambda_i g_i(\omega) \geq0~~~~\forall \omega \in \Omega,
\end{aligned}
\end{equation}
the lower prevision of $f$. The upper prevision of $f$ is equal to $\overline{E}(f)=-\underline{E}(-f)$.
\end{definition}

The lower prevision of a gamble is Alice's supremum buying price for $f$, i.e., how much she should pay to buy
the gamble $f$. The upper prevision is Alice's infimum selling price for $f$, i.e., how much she should ask to sell 
the gamble $f$. We will show in Section \ref{subsec:dual} that the lower and upper prevision are just 
 the lower and upper expectation for the gamble $f$.
By exploiting \eqref{eq:posi}--\eqref{eq:NE}, we can equivalently rewrite \eqref{eq:lp} as:
\begin{equation}
\label{eq:lpne0}
\begin{aligned}
 \underline{E}(f)=&\sup_{\gamma_0\in \reals, \lambda_i \in \reals^+} \gamma_0\\
 &s.t:\\
 &f - \gamma_0\mathbbm{1}  -\sum\limits_{i=1}^{|\mathcal{G}|} \lambda_i g_i \in \nonnegative,
\end{aligned}
\end{equation}
or equivalently,
\begin{equation}
\label{eq:lpne1}
\begin{aligned}
 \underline{E}(f)=&\sup_{\gamma_0\in \reals} \gamma_0 ~~~s.t.~~~ f- \gamma_0\mathbbm{1}\in N(\assess).
\end{aligned}
\end{equation}
In other words, we have expressed the constraint in the above optimisation problems as a membership.
\begin{example}
 Let us consider again the coin example and the set  of  assessments $\assess=\{g_1=(-1,1),g_2=(1,-1)\}$.
 It can be verified that $N(\assess)$ coincides with the maximal closed convex cone in Figure~\ref{fig:coin0}(d).
In this case, the lower prevision for the gamble $f=(0,1)$ is $\underline{E}(f)=\frac{1}{2}$ and it is equal to the upper prevision.
For maximal coherent set of desirable gambles, lower and upper previsions always coincide.
If Alice had accepted only the gambles $\assess=\{g_1=(-0.1,1),g_2=(1,-0.1)\}$ resulting in the closed convex cone of Figure~\ref{fig:coin0}(c), then the lower prevision for the gamble $f$ would be $\underline{E}(f)\approx 0.09$ and the upper prevision $\overline{E}(f)\approx 0.91$.
\end{example}

Having defined  lower and upper previsions, we can better understand \ref{eq:sl}.
\ref{eq:sl} can be formulated as the following decision problem

  \begin{equation}
\label{eq:feasibilBob}
\begin{array}{l}
\exists \lambda_i\geq0  ~s.t.~ \sum\limits_{i=1}^{|G|}\lambda_ig_i(\omega) < 0, ~~~~\forall \omega \in \Omega,
\end{array}
\end{equation} 
there exists a combination of Alice's desirable gambles that is negative.
Let us assume such $\lambda^*_i$ exist, that is $\sum_{i=1}^{|G|}\lambda^*_ig_i<0$. 
Then another agent, Bob, could sell to Alice the gambles $\lambda^*_ig_i$ and she would accept them because
$g_i$ is desirable to her and so $\lambda^*_ig_i$ (by \ref{eq:NE}).
Overall  Bob would give away $-\sum_{i=1}^{|G|}\lambda^*_ig_i$.
However, since $-\sum_{i=1}^{|G|}\lambda^*_ig_i>0$ $\forall \omega \in \Omega$, he  actually 
gains utiles no matter the result of the experiment. 
Bob's gain is equivalent to Alice's loss ($\sum_{i=1}^{|G|}\lambda^*_ig_i$), hence Alice 
can be used as a money pump.\footnote{By \ref{eq:NE}, Alice would also accept the gambles $\gamma\lambda^*_ig_i$
for  $\gamma>1$ allowing Bob to multiply his gain of $\gamma$.}
In Economics, such situation is called an \textit{arbitrage}, while in the subjective definition of probability  is called
a \textit{Dutch book}.

Hence, finally we notice that, by Equation \eqref{eq:lpne0} and Proposition \ref{prop:coco1}, the problem of checking whether $\domain$ is coherent (the coherence problem) can be formulated as the following decision problem:
\begin{equation}
\label{eq:dec}
\begin{aligned}
\exists \lambda_i\geq0:-\mathbbm{1}-\sum\limits_{i=1}^{|\mathcal{G}|} \lambda_i g_i \in \gambles^{\geq}.
\end{aligned}
\end{equation}
If the answer is ``yes'', then the gamble $-\mathbbm{1}$ belongs to $\domain$, proving $\domain$'s incoherence.  The coherence problem  therefore also reduces to the problem of evaluating the nonnegativity of a function in the considered space (let us call this problem the ``nonnegativity decision problem'').


\subsection{Probabilistic interpretation thorough duality}\label{subsec:dual}
The aim of this Section is to  provide a natural probabilistic interpretation to the theory of desirable gambles $\theory$.
This is done by  showing a stronger result, namely that the \textit{dual} of a coherent set of desirable gambles is a closed convex set
of probability charges:
\begin{equation}
\mathcal{P}=\left\{\mu\in \mathcal{M}^{\geq}: \int \mathbbm{1} d\mu=1,~~\int gd\mu \geq0,~ ~\forall g \in \domain \right\},
\end{equation}
where $\mathcal{M}^{\geq}$ is the set of nonnegative charges. Observe that the term ``charge'' is used in Analysis to denote a finitely additive set function \citep[Ch.11]{aliprantisborder}.
Conversely a measure is a countably additive set function. In this paper we use charges to be more general, but this does not really matter
for the results about QT that we are going to present later on.

The key point in the duality proof  is that $\gambles^{\geq}$ (the set of all nonegative gambles (real-valued bounded function) on $\pspace$)
includes indicator functions.\footnote{An indicator function defined on a subset $A\subseteq \pspace$ is a function that is equal to one 
for all elements in $A$ and $0$ for all elements outside $A$.}
This is crucial to prove that the dual of $\domain$ is always included in $\mathcal{M}^{\geq}$. We will see in the next sections that when this is not the case,
the dual of a coherent set of desirable gambles is not anymore a convex set of probabilities.

Note that, equipped with the supremum norm, $\gambles$  constitutes a Banach space, and its topological dual $\mathcal{L}^*$ is the space  of all bounded  functionals on it. We assume the weak${}^*$ topology on $\mathcal{L}^*$.

Let $\mathcal{A}$ be the algebra of subsets of $\Omega$  and $\mu:\mathcal{A} \rightarrow [-\infty,\infty]$ 
denotes a charge: that is $\mu$ is a  finitely additive set function of $\mathcal{A}$  \citep[Ch.11]{aliprantisborder}, \citep{bhaskara1983},
that can take positive and negative values.
We have that 
 every gamble on $\mathcal{L}$ 
 is integrable with respect to any finite charge \citep[Th.11.8]{aliprantisborder}. 
Therefore, for any gamble $g$ and finite charge $\mu$ we can define  $\int gd\mu$, which we can interpret as
a linear functional  $L(\cdot):=\langle \cdot, \mu\rangle$ on  
$g$. 
We denote by $\mathcal{M}$ the set of all finite charges on $\mathcal{L} $ and  by $\mathcal{M}^{\geq}$ the set of nonnegative charges. 
 $\mathcal{M}$ is isometrically isomorphic to $\mathcal{L}^*$. The duality bracket between  $\mathcal{L} $ and $\mathcal{M}$ is given by $\langle f, \mu \rangle := \int f d\mu$, with $f \in \mathcal{L} $ and $\mu \in \mathcal{M}$. 

A linear functional $L$ of gambles is said to be \emph{nonnegative} whenever it  satisfies : $L(g) \geq 0$, for $g \in \nonnegative$. 
A nonnegative linear functional is called a \emph{state} if moreover it preserves the unitary constant gamble. 
In our context, this means $L(\mathbbm{1}) = \langle \mathbbm{1}, \mu\rangle=\int \mathbbm{1} d\mu=1$, i.e., the linear functional is scale preserving. Hence,  the set of states $\mathscr{S}$ corresponds to the closed convex set of  all probability charges.

We define the \emph{dual} of a subset $\domain$ of $\mathcal{L}$ as:\\
\begin{equation}
\domain^\bullet=\left\{\mu\in \mathcal{M}: \int gd\mu \geq0, ~\forall g \in \domain \right\}.
\end{equation}

\begin{proposition}
 The dual of $\gambles$ coincides with $\{0\}$, whereas the dual of $\nonnegative$ is  the set of nonnegative charges $\mathcal{M}^{\geq}$.
\end{proposition}
 Since $(\cdot)^\bullet$ is an anti-monotonic operation on the complete lattice of subsets of $\mathcal{L}$, the dual of any coherent set of desirable gambles is a closed convex cone in between those two extremes. Can they be characterised in some way?
It actually turns out that the dual of a coherent set of desirable gambles can  be completely described in terms of a (closed convex) set of states (probability charges). 
More precisely, we have that:\footnote{All proofs can be found in the Appendix}

\begin{theorem}
\label{prop:dualcharges0}
The map
\[\domain \mapsto  \mathcal{P}:=\domain^\bullet \cap \mathscr{S} \]
establishes a bijection between coherent sets of desirable gambles and non-empty closed convex sets of states.
\end{theorem}

This means that we can write the dual of $\domain$ as the set
\begin{equation}
\label{eq:ddual}
 \mathcal{P}=\left\{\mu\in \states: L_\mu(g) \geq0,~\forall g \in \domain\right\},
\end{equation} 
which is a closed convex-set of probability charges.  We have derived the axioms of probability---a non-negative function that integrates to one---from the the coherence postulate \ref{eq:sl}. 
Hence, as we are going to see at the end of this subsection, whenever an agent is coherent, Equation \eqref{eq:ddual} states that desirability corresponds to non-negative expectation (for all probabilities in $\mathcal{P}$). When $\domain$ is incoherent,  $\mathcal{P}$ turns out to be empty---there is no probability compatible with the assessments in $\domain$.
It is thus form this perspective that it has to be understood the claim that expression \ref{eq:sl} alone captures the coherence postulate as formulate in the introduction in case of classical probability theory, and thus that the latter follows from it. 

As an immediate corollary of the Theorem \ref{prop:dualcharges0}, to say that a closed convex cone is coherent is equivalent to say that its dual is a closed convex subset of states. 
 
As already mentioned, once we have defined the duality between TDG and probability theory, we can immediately reformulate the lower and upper previsions by means of probabilities.
Indeed, let $\assess$ be a finite set of assessments, and $\domain:= N(\assess)$ be a coherent set of
desirable gambles. Given  $f \in \gambles$, the lower prevision of $f$ defined in Equation \eqref{eq:lp} can also be computed as:
\begin{equation}
\label{eq:lpdual}
\begin{aligned}
 \underline{E}(f):=&\inf_{\mu \in \mathscr{S}} \int_{\Omega} f(\omega) d\mu\\
 &s.t:\\
  &\int_{\Omega} g_i(\omega) d\mu\geq0, \forall ~ g_i\in \mathcal{G}.\\
\end{aligned}
\end{equation}
which is equivalent to
\begin{equation}
\label{eq:lpdual1}
\begin{aligned}
 \underline{E}(f):=&\inf_{\mu \in \mathcal{P}} \int_{\Omega} f(\omega) d\mu.
\end{aligned}
\end{equation}
 The upper prevision of $f$ is defined $\overline{E}(f)=-\underline{E}(-f)$.

Hence, the lower and upper prevision of $f$ w.r.t.\ $\domain$ are just the \textit{lower and upper expectation} of $f$
w.r.t.\ $\mathcal{P}$.
In case $\domain$ is maximal, then $\mathcal{P}$ includes only a single probability and, therefore, in this case:
$$
\underline{E}(f)=\overline{E}(f)=E(f).
$$
That is, the solution of \eqref{eq:lpdual} coincides with the expectation of $f$.
We have considered the more general case $\underline{E}(f)\leq \overline{E}(f)$  because, as we will discuss in Section \ref{sec:momentmat}, QT is a theory of ``imprecise'' probability \cite{walley1991}.\footnote{The term ``imprecise'' refers to the fact that the closed convex set $\mathcal{P}$ may not be a singleton, that is the probability may not be ``precisely'' specified. Imprecise probability theory  is also referred as robust Bayesian.}

\section{Taking computational complexity seriously}
\label{sec:comp}
We have seen in Section \ref{sec:primalinference} that the problem of checking whether or not $\domain$ is coherent can be formulated as the following decision problem:
\begin{equation}
\label{eq:dec}
\begin{aligned}
 \exists\lambda_i\geq0:-\mathbbm{1}-\sum\limits_{i=1}^{|\mathcal{G}|} \lambda_i g_i \in \gambles^{\geq}.
\end{aligned}
\end{equation}
If the answer is ``yes'', then the gamble $-\mathbbm{1}$ belongs to $\domain$, proving $\domain$'s incoherence. Moreover, any inference task can ultimately be reduced to a problem of the form~\eqref{eq:dec}, see Section \ref{sec:primalinference}. 
Hence, the above decision problem  unveils  a crucial fact: the hardness of inference in classical probability corresponds to the hardness of evaluating the non-negativity of a function in the considered space (let us call this the ``non-negativity decision problem'').

When $\pspace$ is infinite (in this paper we consider the case  $\pspace \subset\complex^n$) and for generic functions, the non-negativity decision problem is undecidable. To avoid such an issue, we may impose  restrictions on the class of allowed gambles and thus define $\theory$  on a appropriate subspace $\gambles_R$ of $\gambles$.\footnote{The point is that we want $\theory$ defined on $\gambles_R$ to coincide with the restriction to $\gambles_R$ of $\theory$ when defined on $\gambles$. Given this property, we are then assured that the dual of a coherent set $\domain_R$ in $\gambles_R$ can be identified with the dual of its deductive closure in $\gambles$, i.e. with the closed convex set of probability charges $(N(\domain_R))^\bullet$. In Appendix \ref{app:comp} we make the construction and claims precise.} For instance, instead of $\gambles$, we may consider $\gambles_R$: the class of multivariate polynomials of degree at most $d$ (we denote by $\gambles_R^{\geq}\subset\gambles_R$ the subset of non-negative polynomials and by $\gambles_R^{<}\subset\gambles_R$ the negative ones).  In  doing so, by Tarski-Seidenberg quantifier elimination theory \citep{tarski1951decision,seidenberg1954new}, the decision problem becomes decidable, but still intractable, being in general NP-hard. If we  accept the so-called ``Exponential Time Hypothesis'' (that P$\neq$NP) and we require that inference should be tractable  (in P), we are stuck. 
What to do? A solution is to change the meaning of ``being non-negative'' for a  function by considering a subset $\bnonnegative \subsetneq \nonnegative_R$ for which the membership problem in \eqref{eq:dec} is in P.

In other words, a computationally efficient TDG, which we denote by $\btheory$, 
should be based on a logical redefinition of the tautologies, i.e., by stating that
\begin{enumerate}[label=\upshape B0.,ref=\upshape B0]
\item\label{eq:btaut} $\bnonnegative$ should always be desirable,
\end{enumerate}
in the place of~\ref{eq:taut}. The rest of the theory can develop following the footprints of the original theory. In particular, the deductive closure for is $\btheory$ is defined by:
\begin{enumerate}[label=\upshape B1.,ref=\upshape B1]
\item\label{eq:b1} $N_B(\assess)\coloneqq\cl \posi(\bnonnegative \cup \mathcal{G})$.
\end{enumerate}
And sometimes we denote $N_B(\assess)$ by $\bdomain$. Again, $\bdomain = \posi(\bnonnegative \cup \mathcal{G})$ for finite $\assess$.

Finally, the coherence postulate, which now naturally encompasses the computation postulate, states that:
\begin{definition}[P-coherence]
\label{def:bavs}
A set $\bdomain$ of desirable gambles is \emph{P-coherent}  if and only if\
\begin{enumerate}[label=\upshape B2.,ref=\upshape B2]
\item\label{eq:b2} $\bnegative \cap  \bdomain=\emptyset$.
\end{enumerate}
\end{definition}
Above we called P-coherent a set $\bdomain$ that satisfies \ref{eq:b2} since, whenever $\bnonnegative$ contains all positive constant gambles, its incoherence
can be verified in polynomial time by solving:\footnote{For a justification of the non computational part of this claim, see Proposition \ref{prop:cohe}.}
\begin{equation}
\label{eq:bdec}
\begin{aligned}
\exists\lambda_i\geq0 ~~\text{ such that }~~ -\mathbbm{1}_R-\sum\limits_{i=1}^{|\mathcal{G}|} \lambda_i g_i \in \bnonnegative,
\end{aligned}
\end{equation}
where $\mathbbm{1}_R$ denotes the unitary gamble in $\gambles_R$, i.e.,  $\mathbbm{1}_R(\omega)=1$ for all $\omega \in  \Omega$.
 Hence, $\btheory$ and $\theory$ (defined over $\gambles_R$) have  the same deductive apparatus; they just possibly differ in the considered set of tautologies, and thus in their (in)consistencies, as we only ask Alice to always accept gambles for which she can efficiently determine the nonnegativity (P-nonnegative gambles)  and to never accept gambles for which she can efficiently determine the negativity (P-negative gambles).


\subsection{Computationally efficient coherence and its consequences}\label{sec:complex}

Interestingly, we can associate a ``probabilistic'' interpretation as before to the calculus defined by \ref{eq:btaut}--\ref{eq:b2} by computing the dual of a P-coherent set.


Since $\gambles_R$ is a topological vector space, we can consider its dual space $\gambles_R^*$ of all bounded linear functionals $L_B: \gambles_R \rightarrow \reals$. Hence, with the additional condition that linear functional preserves the unitary gamble ($L_B(\mathbbm{1}_R)=1)$,  the dual cone of a P-coherent $\bdomain\subset \gambles_R$  is given by
\begin{equation}
\label{eq:dualL}
 \bdomain^\circ=\left\{L_B \in \gambles_R^* \mid L_B(g)\geq0, ~~L_B(\mathbbm{1}_R)=1,~\forall g \in \bdomain\right\}.
\end{equation}
Analogously to the previous cases, we call \textit{states} the elements of the following closed convex set of linear functionals:
\begin{equation}
 \mathscr{S}_B=\left\{L_B \in \gambles_R^* \mid L_B(g)\geq0, ~~L_B(\mathbbm{1}_R)=1,~\forall g \in \bnonnegative\right\}.
\end{equation}
Hence, we can rewrite the dual $\bdomain^\circ$ as
\begin{equation}
 \bdomain^\circ=\left\{L_B \in \mathscr{S}_B \mid L_B(g)\geq0, ~\forall g \in \bdomain\right\}.
\end{equation}

To $\bdomain^\circ$ we can then associate its extension  $ \bdomain^\bullet$ in $\mathcal{M}$, that is, the set of all  charges
 on $\pspace$ extending an element in $\bdomain^\circ$. In general however this set does not yield a classical probabilistic interpretation to $\btheory$. This is because, whenever
$\bnegative \subsetneq \negative_R$, there are negative gambles that cannot be proved to be negative in polynomial time. 

  \begin{theorem}\label{th:fundamental}
  Assume that $\bnonnegative$ includes all positive constant gambles and that it is closed (in $\gambles_R$). 
Let $\bdomain \subseteq \gambles_R$ be a P-coherent set of  desirable gambles. The following statements are equivalent:
\begin{enumerate}
   \item $\bdomain$ includes a negative gamble that is not in $\bnegative$.
\item $\posi(\nonnegative\cup \mathcal{G})$ is incoherent, and thus $\mathcal{P}$ is empty
    \item $\bdomain^{\circ}$ is not (the restriction to  $\gambles_R$ of) a closed convex set of mixtures of classical evaluation functionals. 
    \item The extension   $ \bdomain^\bullet$ of $\bdomain^{\circ}$ in the space $\mathcal{M}$ of all charges in $\pspace$ includes only signed ones (negative-probabilities).
\end{enumerate}
 \end{theorem}

Theorem \ref{th:fundamental} is the central result of this paper. It states that whenever $\bdomain$ includes a negative gamble (item 1), there is no classical probabilistic interpretation for it (item 2). The other points suggest alternatives solutions to overcome this deadlock: either to change the notion of evaluation functional (item 3) or to use negative-probabilities as a means for interpreting $\btheory$ (item 4).
  
  Let us clarify item 3 above. In doing so, we introduce some terminology. 
Recall that $\mathscr{S}$ is the collection of states (probability charges), that is  of all nonnegative linear functionals on $\gambles$ that preserve the unitary constant gamble.
The extremes of $\mathscr{S}$ are the so-called atomic charges (Dirac's delta), that is the functional $
  e_{\tilde{\omega}} $ assigning $1$ to some given $\tilde{\omega}$ and 0 elsewhere. 
  The linear functional $L(g)= \langle g, e_{\tilde{\omega}}\rangle = g(\tilde{\omega})$ defined by  an atomic charge $e_{\tilde{\omega}}$ is a \emph{classical} evaluation function -- it evaluates the function $g$ at $\tilde{\omega}$.
By Krein-Milman Theorem each state $\mu \in \mathscr{S}$  is a convex combination of atomic charges, or (when the space is not finite) a limit of such combinations. 
 Hence, the linear functional induced by a state $\mu$ is a convex combination (mixture) of classical evaluation functions, or a limit of such combinations.
Recall that any positive functional on $\gambles_R$ can be extended (possibly non uniquely) to a positive functional on $\gambles$ and that the restriction to  $\gambles_R$ of a positive functional on $\gambles$ is a positive functional on $\gambles_R$. 
Hence, since  $ \bdomain^\bullet$  includes charges (that are only affine combinations of classical evaluation functions), its
restriction to $\gambles_R$ cannot be a closed convex set of mixtures of classical evaluation functions.
This is the last statement of the last result.

Implicitly, Theorem \ref{th:fundamental} is also informing us on the properties of $\mathscr{S}_B$. Indeed, the fact that there is a P-coherent set  $\bdomain$ that includes a negative gamble that is not P-negative, yields that $\bnonnegative \subsetneq \nonnegative$, and therefore $(\bnonnegative)^\bullet \supsetneq (\nonnegative)^\bullet$. As a consequence, the extremes of  $\mathscr{S}_B$ are in general only affine combinations of classical evaluation functions.

In the next section, we show that Quantum Theory is a paradigmatic instance of $\btheory$. Given this, Theorem \ref{th:fundamental} applies and it turns out to be
the key  to explains the weirdness of the microscopic world from the perspective of an external classic observer.
However, before doing that, we briefly discuss inference in P-coherent theories.


\subsection{Inference in P-coherent theories}
\label{sec:polyinference}
 In this subsection, we compare inference in theory $\theory^*$ with inference in the classical theory $\theory$.
 
 Let $\assess$ be a finite set of assessment in $\gambles_R$, and $ \bdomain=N_B(\assess)=\posi(\assess \cup \bnonnegative)$ be the corresponding P-coherent set of
desirable gambles in $\theory^*$. The lower prevision of a gamble $f \in \gambles_R$  is defined  as
\begin{equation}
\label{eq:lpber}
\begin{aligned}
 \underline{E}_B(f):=&\sup_{\gamma_0\in \reals, \lambda_i \in \reals^+} \gamma_0\\
 &s.t:\\
 &f - \gamma_0\mathbbm{1}_R  -\sum\limits_{i=1}^{|\mathcal{G}|} \lambda_i g_i \in \bnonnegative.
\end{aligned}
\end{equation}
The upper prevision as $ \overline{E}_B(f)=-\underline{E}_B(-f)$. Comparing \eqref{eq:lpne0} and \eqref{eq:lpber}, the reader can notice that
 $f - \gamma_0\mathbbm{1}  -\sum_{i=1}^{|\mathcal{G}|} \lambda_i g_i \in \nonnegative$ in  \eqref{eq:lpne0}
becomes $f - \gamma_0\mathbbm{1}_R  -\sum_{i=1}^{|\mathcal{G}|} \lambda_i g_i \in \bnonnegative$ in \eqref{eq:lpber}. 

Note that, by definition of $\bnonnegative$,  the membership of $f - \gamma_0\mathbbm{1}_R  -\sum_{i=1}^{|\mathcal{G}|} \lambda_i g_i$ to $\bnonnegative$
can be evaluated in \textbf{polynomial-time} (its complexity class is P). 

Since  $\bnonnegative \neq  \nonnegative \cap \gambles_R$,  $\bdomain$ may not be coherent (in $\theory$), we therefore have that, for every $f \in \gambles_R$
\begin{equation}
\label{eq:belltype}
\stackrel{\substack{\theory^*\\~}}{\underline{E}_B(f)} ~~\leq~~ \stackrel{\substack{\theory\\~}}{\underline{E}(f)} ~~\leq ~~  \stackrel{\substack{\theory\\~}}{\overline{E}(f)}  ~~\leq~~ \stackrel{\substack{\theory^*\\~}}{\overline{E}_B(f)},
\end{equation}
meaning that $\underline{E}_B(f)$ cannot always be interpreted as a lower expectation. 
We claim that \eqref{eq:belltype} is just a general formulation of so-called \textit{Bell-type inequalities} in QT.


\section{Coherence model for a quantum experiment}\label{sec:coheqm}
The aim of this section is to write down the  gambling system for a quantum mechanics experiment.
Since in QT any real-valued quantum observable  is described by a Hermitian operator,
this naturally defines a vector subspace of gambles  $\gambles_R$.
We will then 
show that evaluating the nonnegativity of a gamble in $\gambles_R$  is not computationally efficient.
This will  lead us to define a P-coherence postulate in $\gambles_R$ and, thus, via duality, to derive the first postulate of QT 
 \begin{quote}
\textit{  Associated to any isolated physical system is a complex vector space
with inner product (that is, a Hilbert space) known as the state space of the
system. The system is completely described by its density operator, which is a
positive operator $\rho$ with trace one, acting on the state space of the system.}
 \end{quote}

In the last subsection we thus briefly discuss the case of all other axioms, and how to derive them.


\subsection{Space of gambles in QT}\label{sub:space}
Consider first a single particle system with $n$-degrees of freedom and  let $\overline{\complex}^{n}$ be the $n$-dimensional complex unit-sphere, i.e.,:
$$
\overline{\complex}^{n}\coloneqq\{ x\in \complex^{n}: ~x^{\dagger}x=1\}.
$$
We can interpret an element $\tilde{x} \in \overline{\complex}^{n}$ as ``input data'' for some classical preparation procedure.
For instance,  in the case of the spin-$1/2$ particle ($n = 2$),  if $\theta = [\theta_1 , \theta_2 , \theta_3 ]$ is the direction of a filter
in the Stern-Gerlarch experiment, then $\tilde{x}$ is its one-to-one mapping into $\overline{\complex}^{2}$ (apart from a phase term).
For spin greater than $1/2$, the variable $\tilde{x} \in \overline{\complex}^{n}$ associated to the preparation procedure
 cannot directly be interpreted in terms only of ``filter direction''.
Nevertheless,  at least on the formal level, $\tilde{x}$ plays the role of  a ``hidden variable'' in our model.
Two vectors $\tilde{x},\tilde{x}'$ correspond to the same preparation procedure if
$\tilde{x}'=\delta \tilde{x}$ with $\delta \in \complex$  and $|\delta| = 1$ (phase).
$\overline{\complex}^{n}$ will therefore plays the role of the phase space (the possibility space $\pspace$) and $\tilde{x}$ of the 
``hidden-variable''. This hidden-variable model for QT is also discussed in \cite[Sec.~1.7]{holevo2011probabilistic}, where the author  explains
why this model does not contradict the existing ``no-go'' theorems for hidden-variables, see also Section \ref{sec:hidden1d}.

In QT any real-valued observable  is described by a Hermitian operator.
This naturally imposes restrictions on the type of functions $g$ in \eqref{eq:dec}:
$$
g(x)=x^\dagger G x,
$$
where  $x \in \pspace$ and $G\in \He^{n\times n}$, with $\He^{n\times n}$ being the set of Hermitian matrices of dimension $n \times n$.
Since $G$ is Hermitian and  $x$ is bounded ($x^{\dagger}x=1$), $g$ is a real-valued bounded function (a gamble).
By using the bra-ket notation, a gamble is thus $g(x)=\expval{G}{x}$.

As before, Alice's acceptance of a gamble  depends on her beliefs (uncertainty) about the preparation procedure.

More generally, we can consider composite systems of $m$ particles each one with $n_j$ degrees of freedom.
The corresponding possibility space is the cartesian product of the systems
$$
\pspace=\prod_{j=1}^m \overline{\complex}^{n_j},
$$
whereas  gambles are of the form
\begin{equation}
 \label{eq:gamble_many}
 g(x_1,\dots,x_m)=(\otimes_{j=1}^m x_j)^\dagger G (\otimes_{j=1}^m x_j),
\end{equation}
with $G \in \He^{n \times n}$, $n=\prod_{j=1}^m n_j$ and where $\otimes$ denotes the tensor product between vectors, seen as column matrices.

The justification for the use of the tensor product in composite systems is the following (for a more in depth discussed see  Section \ref{sec:tensorproduct}).
In the theory of desirable gambles, structural judgements such as independence, corresponds to the product of gambles on the single variables. In the specific case we are considering, 
they have the form $\prod_{j=1}^m x_j^{\dagger}  G_j x_j$.
Now, it is not difficult to verify that such product is mathematically  the same as  $(\otimes_{j=1}^m x_j)^\dagger (\otimes_{j=1}^m G_j) (\otimes_{j=1}^m x_j)$. 
By closing the set of product gambles  under the operations of addition and scalar (real number) multiplication, 
we get the vector space whose domain coincide with the collection of gambles of the form as in  \eqref{eq:gamble_many}. 
Hence,  in the setting under consideration, the tensor product is ultimately a derived notion, not a primitive one.

For $m=1$ (a single particle), evaluating the non-negativity of the quadratic form $x^\dagger G x$ boils down to checking whether the matrix $G$ is Positive Semi-Definite (PSD) and therefore can be performed in polynomial time.
This is no longer true for $m\geq 2$: indeed,  in this case there exist polynomials of type \eqref{eq:gamble_many} that are non-negative, but whose matrix $G$ is indefinite (it has at least one negative eigenvalue).
Moreover, it turns out that problem \eqref{eq:dec} is not \emph{tractable}:

\begin{proposition}
The problem of checking the nonnegativity of $g(x_1,\dots,x_m)$ in \eqref{eq:gamble_many} is NP-hard for $m\geq2$. 
\end{proposition}
This result was proven by \cite{gurvits2003classical} for Hermitian biquadratic forms and by \cite{ling2009biquadratic} in the real valued case.

\begin{example}
Consider the following polynomial $g(x_1, x_2)$ of $m=2$ complex variables of dimension $n=2$, $x_1=[x_{11},x_{12}]^T$ and $x_2=[x_{21},x_{22}]^T$
$$
\begin{aligned}
 2 x_{11} x_{11}^{\dagger} x_{22} x_{22}^{\dagger} -  x_{11} x_{12}^{\dagger} x_{21} x_{22}^{\dagger} -  x_{11} x_{12}^{\dagger} y_1^{\dagger} y_{2} -  x_{11}^{\dagger} x_{12} x_{21} x_{22}^{\dagger} - x_{11}^{\dagger} x_{12} y_1^{\dagger} y_{2} + 2 x_{12} x_{12}^{\dagger} x_{21} y_1^{\dagger}.
 \end{aligned}
$$
We have that $g(x_1, x_2)$ is nonnegative in $\overline{\complex}^4$ (it will be verified later), but it cannot be written as $(x_1 \otimes x_2)^\dagger G (x_1 \otimes x_2) $  with $G \geq 0$ (PSD).
\end{example}


\subsection{QT as computational rationality}
\label{sec:polycoher}
We have seen that the problem of checking the nonnegativity of a quadratic forms is in general  a NP-hard problem. 
What to do? As discussed previously, a solution is to change the meaning of ``being non-negative''   by considering a subset $\bnonnegative \subsetneq \nonnegative$ for which the membership problem, and thus \eqref{eq:dec}, is in P. 

For  functions of type  \eqref{eq:gamble_many}, we can extend the notion of non-negativity that holds for a single particle 
to $m>1$ particles:
\begin{enumerate}[label=\upshape,ref=\upshape $\Sigma^{\geq}$]
\vspace{-0.5cm}\item\label{eq:b0} $$\Sigma^{\geq}\coloneqq\{g(x_1,\dots,x_m)=(\otimes_{j=1}^m x_j)^\dagger G (\otimes_{j=1}^m x_j): G\geq0\}.$$
\end{enumerate}
That is, the function is ``non-negative'' (P-nonnegative) whenever $G$ is PSD.  P-nonnegative gambles are also called  \textit{Hermitian sum-of-squares} (see e.g. \cite{d2009polynomial}).  
We will discuss about sum-of-squares polynomials in Sections \ref{sec:HermSOS} and \ref{sec:sos}.
Observe that, in $\Sigma^{\geq}$ the  non-negative constant functions have the form 
$$
g(x_1,\dots,x_m)=c \mathbbm{1}_R(x_1,\dots,x_m)
$$ 
with $c\geq0$ and 
$$
\mathbbm{1}_R(x_1,\dots,x_m)=(\otimes_{j=1}^m x_j)^\dagger I (\otimes_{j=1}^m x_j),
$$ 
being the unitary gamble.

Similarly, a gamble $g$ is P-negative whenever $G$ is Negative-Definite (ND), that is:
\begin{align}
 \Sigma^{<}&=\{g(x_1,\dots,x_m)=(\otimes_{j=1}^m x_j)^\dagger G (\otimes_{j=1}^m x_j): G<0\}.
\end{align}
We therefore can formulate postulates \ref{eq:btaut}-\ref{eq:b2}, and thus in particular  that 
$\bdomain$  is a P-coherent set of desirable gambles whenever $\bnegative \cap  \bdomain=\emptyset$.

In the following subsections, we are thence going to show how QT can be derived from this ``computational rationality''  model and how all paradoxes of QT can be explained as a consequence of P-coherence. Hence again, since both~\ref{eq:b1},\ref{eq:b2} and~\ref{eq:NE},\ref{eq:sl} are the same logical postulates parametrised by the appropriate meaning of ``being negative/non-negative'', the only axiom truly separating classical probability theory from the quantum one is~\ref{eq:btaut} (with the specific form of~\ref{eq:b0}), thus implementing the requirement of computational efficiency.


\subsubsection{Duality}
\label{sec:dualQM}
Recall from Section \ref{sec:complex} that the set $\left\{L_B \in \gambles_R^* \mid L_B(g)\geq0, ~~L_B(\mathbbm{1}_R)=1,~\forall g \in \bdomain\right\}$ is the dual of $\bdomain\subset \gambles_{R}$.

The monomials $\otimes_{j=1}^m x_j$ form a basis of the space $\gambles_{R}$. Define the Hermitian matrix of scalars
\begin{equation}
\label{eq:linearoperator}
Z:=L_B\left((\otimes_{j=1}^m x_j)(\otimes_{j=1}^m x_j)^\dagger\right),
\end{equation} 
and let $\{z_{ij}\}\in \complex^{d}$, with $d=\frac{n(n+1)}{2}$ and $n=\prod_{j=1}^m n_j$, be the vector of variables obtained  by taking  the elements 
of the upper triangular part of $Z$.
Given any gamble $g$, we can therefore rewrite $L_B(g)$ as a function of the vector  $\{z_{ij}\}\in \complex^{d}$. This means that the dual space $\gambles_{R}^*$ is isomorphic to
$\complex^{d}$, and we can thence define the dual maps $(\cdot)^\circ$ between $\gambles_{R}$ and $\complex^{d}$ as follows.
\begin{definition}\label{def:dual0}
Let $\bdomain$ be a closed convex cone in $\gambles_{R}$. Its dual cone is defined as 
\begin{equation}
\label{eq:dualM0}
\bdomain^\circ=\left\{{z} \in \complex^{{d}}:  L_B(g)\geq0, ~\forall g \in \bdomain\right\},
\end{equation} 
where $L_B(g)$ is completely determined by $\{z_{ij}\}$ via the definition \eqref{eq:linearoperator}.
\end{definition}

\begin{example}
 Consider the case $n=m=2$, then
 \begin{equation}
\label{eq:exlinearoperator}
L_B\left((\otimes_{j=1}^2 x_j)(\otimes_{j=1}^2 x_j)^\dagger\right)= L_B\left(\left[\begin{smallmatrix}x_{11} x_{11}^{\dagger} x_{21} x_{21}^{\dagger} & x_{11}^{\dagger} x_{12} x_{21} x_{21}^{\dagger} & x_{11} x_{11}^{\dagger} x_{21}^{\dagger} x_{22} & x_{11}^{\dagger} x_{12} x_{21}^{\dagger} x_{22}\\x_{11} x_{12}^{\dagger} x_{21} x_{21}^{\dagger} & x_{12} x_{12}^{\dagger} x_{21} x_{21}^{\dagger} & x_{11} x_{12}^{\dagger} x_{21}^{\dagger} x_{22} & x_{12} x_{12}^{\dagger} x_{21}^{\dagger} x_{22}\\x_{11} x_{11}^{\dagger} x_{21} x_{22}^{\dagger} & x_{11}^{\dagger} x_{12} x_{21} x_{22}^{\dagger} & x_{11} x_{11}^{\dagger} x_{22} x_{22}^{\dagger} & x_{11}^{\dagger} x_{12} x_{22} x_{22}^{\dagger}\\x_{11} x_{12}^{\dagger} x_{21} x_{22}^{\dagger} & x_{12} x_{12}^{\dagger} x_{21} x_{22}^{\dagger} & x_{11} x_{12}^{\dagger} x_{22} x_{22}^{\dagger} & x_{12} x_{12}^{\dagger} x_{22} x_{22}^{\dagger}\end{smallmatrix}\right]\right),
\end{equation}
and so
 \begin{equation}
\label{eq:exlinearoperator}
Z=\left[\begin{matrix} z_{11} & z_{12} & z_{13} & z_{14}\\
                 z_{12}^{\dagger} & z_{22} & z_{23} & z_{24}\\
                 z_{13}^{\dagger} & z_{23}^{\dagger} & z_{33} & z_{34}\\
                 z_{14}^{\dagger} & z_{24}^{\dagger} & z_{34}^{\dagger} & z_{44}                 
\end{matrix}\right],
\end{equation}
with $z_{11}=L_B(x_{11} x_{11}^{\dagger} x_{21} x_{21}^{\dagger})$, $z_{12}=L_B(x_{11}^{\dagger} x_{12} x_{21} x_{21}^{\dagger} )$ etc..
\end{example}


In discussing properties of the dual space, we need the following well-known result from linear algebra:
\begin{lemma}\label{lem:TR}
For any $M \in H^{d\times d}$ and $v \in \complex^{d}$, it holds that
\begin{equation}
\label{eq:matrixTR}
 Tr(M (v v^\dagger)) = Tr((v v^\dagger)M) =  v^\dagger M v.
\end{equation}
\end{lemma}

By Lemma \ref{lem:TR} and the definitions of $g$ and $Z$, we obtain the following result.
\begin{proposition}\label{prop:LisTR}
Let $g(x_1,\dots,x_m) = (\otimes_{j=1}^m x_j)^\dagger G (\otimes_{j=1}^m x_j)$ and $G$ Hermitian. Then 
for every $z \in  \complex^{{d}}$, it holds that $L_B(g) = Tr({G} Z)$, 
where $Z$ is defined in \eqref{eq:linearoperator}. 
\end{proposition}

The next lemma states that the only symmetry in the matrix $(\otimes_{j=1}^m x_j)(\otimes_{j=1}^m x_j)^\dagger$ with $x_j \in \complex^{n_j}$ is that 
of being self-adjoint.  
\begin{lemma}
\label{lem:uniqueX}
Consider the matrix
 \begin{align}
  \label{eq:X}
X&=(\otimes_{j=1}^m x_j)(\otimes_{j=1}^m x_j)^\dagger.
 \end{align}
 Let $X_{st}$ denote the $st$-th element of $X$ then for all $t\geq s$ (upper triangular elements) we have that $X_{st}=X_{kl}$ 
 $\forall x_j \in \overline{C}^{n_j}$ iff $k=s$ and $l=t$.
\end{lemma}
We then verify that
\begin{proposition}\label{prop:igno}
Let $\bdomain$ be a P-coherent set of desirable gambles. The following holds:
\begin{equation}
\label{eq:dualM1}
\bdomain^\circ=\left\{{z} \in \complex^{{d}}:  L_B(g)=Tr({G} Z)\geq0, ~Z\geq 0,~\forall g \in \bdomain\right\}.
\end{equation} 
\end{proposition}
\begin{proof}
By P-coherence, $\bdomain$ includes $\Sigma^{\geq}$, which is isomorphic to the closed convex cone of PSD matrices.
We have that
$$
L_B(g)=Tr({G} Z)\geq0  ~~\forall g\in \Sigma^{\geq} \subseteq \bdomain.
$$
From a standard result in linear algebra, see for instance \cite[Lemma 1.6.3]{holevo2011probabilistic}, this implies that $Z \geq0$, i.e., it must be a PSD matrix.
\end{proof}
In what follows, we verify that, analogously to Section \ref{sec:complex}, the dual $\bdomain^\circ$ is completely characterised by a closed convex set of states. But before doing that, we have to clarify what is a state in this context.

 In a P-coherent theory, postulate \ref{eq:taut} is replaced with postulate \ref{eq:btaut}. Hence, to define what  a state is, one cannot anymore refer to nonnegative gambles but to gambles that are P-nonnegative. This means that states are linear operators that assign nonnegative real numbers to P-nonnegative, and that additionally preserve the unit gamble.
In the context of Hermitian gambles, the unitary gamble is  
\begin{equation}
\label{eq:constntQM}
\mathbbm{1}_R(x_1,\dots,x_m)=(\otimes_{j=1}^m x_j)^\dagger I (\otimes_{j=1}^m x_j)= \prod\limits_{i=1}^m {x_j}^{\dagger}x_j=1,
\end{equation} 
where $I$ is the identity matrix.
Therefore, we want that
\begin{equation}
\begin{aligned}
L_B\Big((\otimes_{j=1}^m x_j)^\dagger I (\otimes_{j=1}^m x_j)\Big)&=L_B\Big(Tr\Big( I ~(\otimes_{j=1}^m x_j)(\otimes_{j=1}^m x_j)^\dagger\Big)\Big)\\
&=Tr\Big(I~ L_B\Big( (\otimes_{j=1}^m x_j)(\otimes_{j=1}^m x_j)^\dagger)\Big)\Big)\\
&=Tr(I ~Z)=Tr(Z)=1.
\end{aligned}
\end{equation}
Hence, the set of states is
\begin{equation}
\begin{aligned}
 \mathscr{S}_B=\{{z} \in \complex^{{d}}: \mid Z \geq0, ~~Tr(Z)=1\}.
\end{aligned}
\end{equation}

By reasoning exactly as for Theorem \ref{prop:dualcharges0}, we then have the following result.
\begin{theorem}
\label{th:dualSOS}
The map
\[ \bdomain \mapsto \mathcal{Q}:=\bdomain^\circ\cap \mathscr{S}_B\]
is a bijection between P-coherent set of desirable gambles in $\gambles_R$  and closed convex subsets of $\mathscr{S}_B$. 
\end{theorem}

Hence, we can identify the dual of a P-coherent set of desirable gambles $\bdomain$, with the closed convex set of states
\begin{equation}
\label{eq:dualM1}
\begin{aligned}
\mathcal{Q}
&=\left\{ z \in  \mathscr{S}_B:  L_B(g)=Tr({G} Z)\geq0,  ~~\forall g \in \bdomain\right\},
\end{aligned}
\end{equation} 
which is equivalent to .

Notice that the matrices corresponding to states are density matrices, in fact \eqref{eq:dualM1} is equivalent to
$\left\{ \rho \in  \He^{n\times n}: \rho\geq0, Tr(\rho)=1, Tr({G} \rho)\geq0,  ~~\forall G  \in \bdomain\right\}$.
Hence, from now we can identify the set $\mathscr{S}_B$  with the set of density matrices and denote its elements as usual with the symbol $\rho$.



\subsubsection{Inference in QT}
\label{sec:QTinference}
 In this subsection, we use the results of Section \ref{sec:polyinference} to derive lower 
 and upper previsions  of a gamble $f(x_1,\dots,x_m) = (\otimes_{j=1}^m x_j)^\dagger F (\otimes_{j=1}^m x_j)$ in QT. 
 By \eqref{eq:lpber}, we have that
\begin{equation}
\label{eq:lpberqt}
\begin{aligned}
 \underline{E}_B(f):=&\sup_{\gamma_0\in \reals, \lambda_i \in \reals^+} \gamma_0\\
 &s.t:\\
 &(\otimes_{j=1}^m x_j)^\dagger F (\otimes_{j=1}^m x_j) - \gamma_0\mathbbm{1}_R(x_1,\dots,x_m)  -\sum\limits_{i=1}^{|\mathcal{G}|} \lambda_i (\otimes_{j=1}^m x_j)^\dagger G_i (\otimes_{j=1}^m x_j) \in \bnonnegative,\\
\end{aligned}
\end{equation}
the upper prevision  $ \overline{E}_B(f)=-\underline{E}_B(-f)$. 
Note that the membership problem in \eqref{eq:lpberqt} reduces to find $\gamma_0\in \reals, \lambda_i \in \reals^+$ such that
$F- \gamma_0I -\sum\limits_{i=1}^{|\mathcal{G}|}G_i$ is PSD.
The dual of the above optimisation problem is
\begin{equation}
\label{eq:dualpberqt}
\begin{aligned}
 \underline{E}_B(f):=&\inf_{\rho\in \mathscr{S}_B} Tr(F\rho)\\
 &s.t:\\
 &Tr(G_i \rho) \geq0,~~\forall i=1,\dots,|\mathcal{G}|
\end{aligned}
\end{equation}
Whenever a P-coherent $\bdomain$ is maxima, its dual $\mathcal{Q}$ includes a single density matrix $\tilde{\rho}$\footnote{This does not require $|\mathcal{G}|=\infty$.
Since $\tilde{\rho}$ is a matrix it can be uniquely specified by finite assessments of desirability $\tilde{\rho}=\{\rho: Tr(G_i \rho) \geq0,~~\forall i=1,\dots,|\mathcal{G}|\}$.}.
$$
 \underline{E}_B(f) = \overline{E}_B(f)=Tr(F\tilde{\rho}).
 $$
Hence, in this case both the lower and upper prevision of the gamble $f(x_1,\dots,x_m) = (\otimes_{j=1}^m x_j)^\dagger F (\otimes_{j=1}^m x_j) $
 coincide with $Tr(F\tilde{\rho})$.  
 Note that in QT the expectation of $f$ is $Tr(F \rho)$. This follows by Born's rule, a law giving the probability that a measurement on a quantum system will yield a given result.
The agreement with Born's rule is an important constraint in any alternative axiomatisation of QT. Our theory agrees with it although this is a derived notion in our setting.  In fact, in the  view of a density matrix as a  dual operator, $\rho$ is formally equal to
\begin{equation}
  \label{eq:momepr00}
\rho=L\left((\otimes_{j=1}^m x_j)(\otimes_{j=1}^m x_j)^\dagger\right).
\end{equation}
Hence, when a projection-valued measurement characterised by the projectors $\Pi_1,\dots,\Pi_n$ is considered, then
$$
L( (\otimes_{j=1}^m x_j)^\dagger \Pi_i (\otimes_{j=1}^m x_j))=Tr(\Pi_i  L((\otimes_{j=1}^m x_j)(\otimes_{j=1}^m x_j)^\dagger))=Tr(\Pi_i \rho ).
$$
Since $\Pi_i\geq0$ and the polynomials $(\otimes_{j=1}^m x_j)^\dagger \Pi_i (\otimes_{j=1}^m x_j)$ for $i=1,\dots,n$ form a partition of unity, i.e.:
$$
\sum_{i=1}^n (\otimes_{j=1}^m x_j)^\dagger \Pi_i (\otimes_{j=1}^m x_j)=  (\otimes_{j=1}^m x_j)^\dagger I (\otimes_{j=1}^m x_j)=1,
$$
we have that
$$
Tr(\Pi_i \rho )\in[0,1] \text{ and } \sum_{i=1}^n Tr(\Pi_i \rho )=1.
$$
For this reason,  $Tr(\Pi_i \rho )$ is usually interpreted as a probability. But \emph{the projectors $\Pi_i$'s are not indicator functions}, whence, strictly speaking, the traditional interpretation is incorrect. This can be seen clearly in the special case where postulates~\ref{eq:taut} and~\ref{eq:btaut} coincide, as in the case of a single particle, that is, where the theory can be given a classical probabilistic interpretation, see next section. 
In such a case, the corresponding $\rho$ is just a (\emph{truncated}) \emph{moment matrix}, i.e., one for which there is
at least one probability  such that $E[(\otimes_{j=1}^m x_j)(\otimes_{j=1}^m x_j)^\dagger]=\rho$. In summary, our standpoint here is that $Tr(\Pi_i \rho )$ should rather be interpreted as the expectation of the $m$-quadratic form $(\otimes_{j=1}^m x_j)^\dagger \Pi_i (\otimes_{j=1}^m x_j)$. This makes quite a difference with the traditional interpretation since in our case there can be (and usually there will be) more than one charge compatible with such an expectation, as we will point out more precisely, in next section.

\subsection{Truncated moment matrices vs.\ density matrices}
\label{sec:momentmat}
In a single particle system of dimension $n$, $\rho=L_B(x x^{\dagger})$.
 In this case, $ \rho$ can be interpreted as a truncated moment matrix, i.e., there exists a probability distribution on the complex vector variable $x\in\Omega$ such that
 \begin{equation}
  \label{eq:ccc}
   \rho=\int_{x\in \Omega} xx^{\dagger} d\mu(x).
 \end{equation}
In fact, consider the eigenvalue-eigenvector decomposition of the density matrix:
$$
\rho=\sum\limits_{i=1}^n \lambda_i v_iv_i^{\dagger},
$$
with $\lambda_i\geq0$ and $v_i \in \complex^{n}$  being orthonormal.
We can define the probability distribution
$$
\mu(x)= \sum\limits_{i=1}^n \lambda_i \delta_{v_i}(x),
$$
where $\delta_{v_i}$ is an atomic charge (Dirac's delta) on $v_i$.
Then it is immediate to verify that 
$$
\int_{x\in \Omega} x x^{\dagger}d\mu(x)=\sum\limits_{i=1}^n \lambda_i v_iv_i^{\dagger}=\rho.
$$
Note also that,  a truncated moment matrix does not uniquely define a probability distribution, i.e., for a given $\rho$ there may exist two probability distributions $\mu_1(x)\neq\mu_2(x)$ such that
$$
 \rho=\int_{x\in \Omega} xx^{\dagger} d\mu_1(x)=\int_{x\in \Omega} xx^{\dagger} d\mu_2(x).
$$
This means that if we interpret  $\rho$ as a truncated moment matrix, QT is a theory of imprecise probability \cite{walley1991}, that
is a density matrix defines a closed convex set of probability distributions via \eqref{eq:ccc}.
For $m>1$ particles, $\rho$ can be interpreted as a truncated moment matrix only when the P-coherent
set of desirable gambles associated to $\rho$, i.e., $\bdomain=\{(\otimes_{j=1}^m x_j)^{\dagger} G (\otimes_{j=1}^m x_j): Tr(G\rho)\geq0\}$, does not satisfy the first condition of Theorem \ref{th:fundamental}. We will discuss more on that in the next sections. 

%



\section{Explaining the weird}
\label{sec:entan0}
\subsection{Entanglement}\label{sec:entang}

Entanglement is usually presented as a  characteristic of QT. In this section we are going to show that it is actually an immediate consequence of computational tractability, meaning that entanglement phenomena are not confined to QT but can be observed in other contexts too. An example of a non-QT entanglement is provided  in Section  \ref{sec:ent_not_only}.

To illustrate the emergence of entanglement from P-coherence, we verify  that  the set of desirable gambles whose dual is an entangled density matrix $\rho_{e}$ includes a negative gamble that is not in $\bnegative$, and thus, although
being logically coherent, it cannot be given a classical probabilistic interpretation.

In what follows we focus only on bipartite systems $\pspace_A \times \pspace_B$, with $n=m=2$. The results are nevertheless general.

Let $(x,y) \in \pspace_A \times \pspace_B$, where
$x=[x_1,x_2]^T$ and $y=[y_1,y_2]^T$. 
We aim at showing that there exists a gamble $h(x,y)=(x \otimes y)^{\dagger} H (x \otimes y) $ satisfying:
\begin{equation}
\label{eq:condepr0}
\begin{aligned}
 Tr(H \rho_{e})&\geq0 \text{ and }\\
h(x,y)=(x \otimes y)^{\dagger} H (x \otimes y) &< 0 \text{ for all } (x,y)\in \pspace_A\times \pspace_B.\\
\end{aligned}
\end{equation}
The first inequality says that $h$ is desirable in $\btheory$. That is, $h$ is a gamble desirable to Alice whose beliefs are represented by $\rho_{e}$. The second inequality says that $h$ is negative and, therefore, leads to a sure loss in $\theory$.
By~\ref{eq:btaut}--\ref{eq:b2}, the inequalities in~\eqref{eq:condepr0} imply that $H$ must be an indefinite Hermitian matrix.

 Assume that $n=m=2$ and consider the entangled density matrix: 
$$
\rho_{e}=\frac{1}{2}\begin{bmatrix}
      1 & 0 & 0 &1\\
      0 & 0 & 0 &0\\
      0 & 0 & 0 &0\\
      1 & 0 & 0 &1\\
     \end{bmatrix},
$$
and the Hermitian matrix:
$$
H=\left[\begin{matrix}0.0 & 0.0 & 0.0 & 1.0\\0.0 & -2.0 & 1.0 & 0.0\\0.0 & 1.0 & -2.0 & 0.0\\1.0 & 0.0 & 0.0 & 0.0\end{matrix}\right].
$$
This matrix is indefinite (its eigenvalues are $\{1, -1, -1, -3\}$) and is such that $Tr(H\rho_{e})=1$.
Since $Tr(H\rho_{e})\geq0$, the gamble
\begin{equation}
\label{eq:sosqt}
\begin{aligned}
(x \otimes y)^{\dagger} H (x \otimes y)=&
- 2 x_{1} x_1^{\dagger} y_{2} y_2^{\dagger} +  x_{1} x_2^{\dagger} y_{1} y_2^{\dagger} +  x_{1} x_2^{\dagger} y_1^{\dagger} y_{2} +  x_1^{\dagger} x_{2} y_{1} y_2^{\dagger} +  x_1^{\dagger} x_{2} y_1^{\dagger} y_{2} - 2 x_{2} x_2^{\dagger} y_{1} y_1^{\dagger},
 \end{aligned}
\end{equation}
is desirable for Alice in $\btheory$. 

Let  $x_i=x_{ia}+\iota x_{ib}$ and $y_i=y_{ia}+\iota y_{ib}$ with $x_{ia},x_{ib},y_{ia},y_{ib}\in \reals$, for $i=1,2$, denote the real and imaginary
components of $x,y$. Then 
\begin{equation}
\label{eq:realim}
\begin{aligned}
(x \otimes y)^{\dagger} H (x \otimes y)=&- 2 x_{1a}^{2} y_{2a}^{2} - 2 x_{1a}^{2} y_{2b}^{2} + 4 x_{1a} x_{2a} y_{1a} y_{2a} + 4 x_{1a} x_{2a} y_{1b} y_{2b} \\
&- 2 x_{1b}^{2} y_{2a}^{2} - 2 x_{1b}^{2} y_{2b}^{2} + 4 x_{1b} x_{2b} y_{1a} y_{2a} + 4 x_{1b} x_{2b} y_{1b} y_{2b}\\
&- 2 x_{2a}^{2} y_{1a}^{2} - 2 x_{2a}^{2} y_{1b}^{2} - 2 x_{2b}^{2} y_{1a}^{2} - 2 x_{2b}^{2} y_{1b}^{2}\\
=&-(\sqrt{2}x_{1a}y_{2a}-\sqrt{2}x_{2a}y_{1a})^2-(\sqrt{2}x_{1a}y_{2b}-\sqrt{2}x_{2a}y_{1b})^2\\
&-(\sqrt{2}x_{1b}y_{2b}-\sqrt{2}x_{2b}y_{1b})^2
-(\sqrt{2}x_{2b}y_{1a}-\sqrt{2}x_{2a}y_{1b})^2<0.
 \end{aligned}
\end{equation}

This is the essence of the quantum puzzle:  $\bdomain$ is P-coherent but (Theorem~\ref{th:fundamental}) there is no $\mathcal{P}$ associated to it and therefore, from the point of view of a classical probabilistic interpretation, it is not coherent (in any classical description of the composite quantum system, the variables $x,y$ appear to be entangled in a way unusual for classical subsystems). 

As previously mentioned, there are two possible ways out from this impasse:
  to claim the existence of either non-classical evaluation functionals or of negative probabilities. Let us examine them in turn.

\begin{description}
\item[(1) Existence of non-classical evaluation functionals:]
From an informal betting perspective, the effect  of a quantum experiment on $h(x,y)$ is to evaluate this polynomial to return the payoff for Alice. By Theorem~\ref{th:fundamental}, there is no compatible classical evaluation functional, and thus in particular no 
value of the variables $x,y\in \pspace_A \times  \pspace_B$, such that 
$ h(x,y)=1$. Hence, if we adopt this point of view, we have to find another, non-classical, explanation for $ h(x,y)=1$. The following evaluation functional, denoted as $ev(\cdot)$, may do the job:
$$
\text{ev}\hspace{-1.1mm}\left(\begin{bmatrix}
x_1y_1\\
x_2y_1\\
x_1y_2\\
x_2y_2\\
\end{bmatrix}\right)=\begin{bmatrix}
\tfrac{\sqrt{2}}{2}\\
0\\
0\\
\tfrac{\sqrt{2}}{2}\\
\end{bmatrix},~\text{which implies}~ \text{ev}\hspace{-1mm}\left((x \otimes y)^{\dagger} H (x \otimes y)\right)=1.
$$

Note that, $x_1y_1=\tfrac{\sqrt{2}}{2}$ and $x_2y_1=0$ together imply that
$x_2=0$, which contradicts $x_2y_2=\tfrac{\sqrt{2}}{2}$.\footnote{If the product of two complex numbers $x_2 y_1$ is zero then either $x_2=0$ or $y_1=0$.} Similarly,
 $x_2y_2=\tfrac{\sqrt{2}}{2}$ and $x_1y_2=0$ together imply that
$x_1=0$, which contradicts $x_1y_1=\tfrac{\sqrt{2}}{2}$. Hence, as expected, the above evaluation functional is non-classical. It amounts to  assigning a value to the products $x_iy_j$ but not to the single components of $x$ and $y$ separately. Quoting \cite[Supplement~3.4]{holevo2011probabilistic}, ``entangled states are holistic entities in which the single components only exist virtually''. 



\item[(2) Existence of negative probabilities:] Negative probabilities are not an intrinsic characteristic of QT. They
appear whenever one attempts to explain QT ``classically'' by looking at the space of charges on $\pspace$.
To see this, consider $\rho_e$, and assume that, based on \eqref{eq:momepr}, one calculates: 
  \begin{equation}
  \label{eq:momepr}
\int  \begin{bmatrix}x_{1} x_1^{\dagger} y_{1} y_1^{\dagger} & x_1^{\dagger} x_{2} y_{1} y_1^{\dagger} & x_{1} x_1^{\dagger} y_1^{\dagger} y_{2} & x_1^{\dagger} x_{2} y_1^{\dagger} y_{2}\\x_{1} x_2^{\dagger} y_{1} y_1^{\dagger} & x_{2} x_2^{\dagger} y_{1} y_1^{\dagger} & x_{1} x_2^{\dagger} y_1^{\dagger} y_{2} & x_{2} x_2^{\dagger} y_1^{\dagger} y_{2}\\x_{1} x_1^{\dagger} y_{1} y_2^{\dagger} & x_1^{\dagger} x_{2} y_{1} y_2^{\dagger} & x_{1} x_1^{\dagger} y_{2} y_2^{\dagger} & x_1^{\dagger} x_{2} y_{2} y_2^{\dagger}\\x_{1} x_2^{\dagger} y_{1} y_2^{\dagger} & x_{2} x_2^{\dagger} y_{1} y_2^{\dagger} & x_{1} x_2^{\dagger} y_{2} y_2^{\dagger} & x_{2} x_2^{\dagger} y_{2} y_2^{\dagger}\end{bmatrix} d \mu(x,y)=\frac{1}{2}\begin{bmatrix}
      1 & 0 & 0 &1\\
      0 & 0 & 0 &0\\
      0 & 0 & 0 &0\\
      1 & 0 & 0 &1\\
     \end{bmatrix}.
\end{equation}
Because of Theorem~\ref{th:fundamental},  there is no probability charge $\mu$ satisfying these moment constraints, the only compatible being signed ones.
Box~1 reports the nine components and corresponding weights of one of them:
\begin{equation}
\label{eq:momsigned}
 \mu(x,y)=\sum\limits_{i=1}^{9}w_i\delta_{\{(x^{(i)},y^{(i)})\}}(x,y) ~~\text{ with }~~ (w_i,x^{(i)},y^{(i)}) ~~\text{ as in Box~1. } 
\end{equation}
Note that some of the weights are negative but $\sum_{i=1}^{9}w_i=1$, meaning that we have an affine combination of atomic charges (Dirac's deltas).

{

\newpage
\begin{framed}

\vspace{-0.5ex}

\noindent {{\bf Box 1: charge compatible with \eqref{eq:momepr}}}

\vspace{-0.9ex}

{\footnotesize
\noindent 
\begin{center}
 \tiny{
\begin{tabular}{|c|c|c|c|c|c|}
\multicolumn{1}{c}{} & \multicolumn{1}{c}{1} & \multicolumn{1}{c}{2} & \multicolumn{1}{c}{3} & \multicolumn{1}{c}{4} & \multicolumn{1}{c}{5}\\
\hline
 \multirow{ 2}{*}{x} & -0.0963 - 0.6352$\iota$ & 0.251 - 0.9665$\iota$ & 0.7884 + 0.2274$\iota$ & 0.5702 - 0.4006$\iota$ & 0.3452 - 0.4539$\iota$\\
 & -0.0065 - 0.7663$\iota$ & 0.0387 + 0.0381$\iota$ & -0.1263 +  0.5574$\iota$ & 0.0027 + 0.7172$\iota$ & 0.4872 + 0.6613$\iota$\\\hline
 \multirow{ 2}{*}{y} &  -0.3727 - 0.3899$\iota$ & 0.6359 - 0.5716$\iota$ & 0.1553 - 0.4591$\iota$ & -0.3515 + 0.2848$\iota$ & 0.2129 - 0.2004$\iota$\\
 & -0.4385 - 0.7189$\iota$ & 0.3725 + 0.3608$\iota$ & 0.4039 +  0.7759$\iota$ & 0.5911 - 0.6678$\iota$ & 0.2032 - 0.9345$\iota$\\ \hline
 w &     0.4805&  0.7459 & -0.892 &  0.7421 &  0.4724 \\ \hline  
 \end{tabular}\\
 \vspace{0.5cm}
  \begin{tabular}{|c|c|c|c|c|}
\multicolumn{1}{c}{} & \multicolumn{1}{c}{6} & \multicolumn{1}{c}{7} & \multicolumn{1}{c}{8} & \multicolumn{1}{c}{9} \\
  \hline
  \multirow{ 2}{*}{x} & 0.818 + 0.2654$\iota$ & -0.0541 - 0.8574$\iota$ & -0.3179 - 0.1021$\iota$ & -0.1255 - 0.3078$\iota$  \\
 & -0.486 + 0.1556$\iota$ & 0.4995 + 0.1112$\iota$ & 0.5198 + 0.7864$\iota$ & 0.2943 - 0.8961$\iota$  \\
 \hline
 \multirow{ 2}{*}{y} &  0.446 + 0.6996$\iota$ & -0.1628 + 0.561$\iota$ & 0.6285 - 0.4852$\iota$ & 0.0933 - 0.4588$\iota$  \\
 & -0.5474 - 0.1096$\iota$ & -0.8105 + 0.0419$\iota$ & -0.0035 - 0.6079$\iota$ & 0.8455 + 0.2568$\iota$  \\
 \hline
 w &      0.3297 &  -0.7999 &        -0.2544 &  0.1755\\
 \hline
\end{tabular}}\end{center}
The table reports the components of a  charge $\sum_{i=1}^{9}w_i\delta_{\{(x^{(i)},y^{(i)})\}}(x,y)$ that satisfies \eqref{eq:momepr}. The $i$-th column
of the row denoted as $x$ (resp. $y$) corresponds to the element $x^{(i)}$ (resp. $y^{(i)}$). The $i$-th column of the vector
$w$ corresponds to $w_i$. 
Consider for instance the first monomial $x_{1} x_1^{\dagger} y_{1} y_1^{\dagger}$ in \eqref{eq:momepr}, its expectation w.r.t.\ the above charge is
{\scriptsize
$$
\begin{aligned}
&\int x_{1} x_1^{\dagger} y_{1} y_1^{\dagger} \left(\sum_{i=1}^{9}w_i\delta_{\{(x^{(i)},y^{(i)})\}}(x,y)\right) dxdy=\sum_{i=1}^{9}w_i x^{(i)}_{1} {x^{(i)}_1}^{\dagger} y^{(i)}_{1} {y^{(i)}_1}^{\dagger}\\
&= 0.4805 (-0.0963 - 0.6352\iota)(-0.0963 + 0.6352\iota)(-0.3727 - 0.3899\iota)(-0.3727 + 0.3899\iota)\\
&+ 0.7459 (0.251 - 0.9665\iota)(0.251 + 0.9665\iota)(-0.1628 + 0.561\iota)(-0.1628 - 0.561\iota)\\
&+\dots\\
&+ 0.1755(-0.1255 - 0.3078\iota)(-0.1255 + 0.3078\iota)(0.0933 - 0.4588\iota)(0.0933 + 0.4588\iota)=\frac{1}{2}.
\end{aligned}
$$}
 }
\end{framed}
}

The  charge described in Box~1  is one among the many that satisfy~\eqref{eq:momepr} and has been derived numerically.
Explicit procedure for constructing such  charge  representations have been developed by \citep{schack2000explicit}.
An optimisation procedure to find the representation with the
minimum amount of negativity is given in \citep{sperling2009necessary}.
%
%
\end{description}

Again, we want to stress that the two above paradoxical interpretations are a consequence of Theorem~\ref{th:fundamental}, and therefore can emerge when considering any instance of a theory of P-coherence in which
the hypotheses of  this result hold.

\subsection{Local realism}\label{sec:local}
The issue of \emph{local realism} in QT arises when one performs measurements on a pair of separated but entangled particles.
This again shows the impossibility of a peaceful agreement between the internal logical consistency of a P-coherent theory and the attempt to provide an external coherent (classical) interpretation.
Let us discuss it from the latter perspective. Firstly, notice that, since $(x_{1} x_1^{\dagger}+x_{2} x_2^{\dagger})=1$, the linear operator
\begin{align}
\label{eq:matL}
L\left(\left[\begin{matrix}x_{1} x_1^{\dagger} y_{1} y_1^{\dagger} & x_1^{\dagger} x_{2} y_{1} y_1^{\dagger} & x_{1} x_1^{\dagger} y_1^{\dagger} y_{2} & x_1^{\dagger} x_{2} y_1^{\dagger} y_{2}\\x_{1} x_2^{\dagger} y_{1} y_1^{\dagger} & x_{2} x_2^{\dagger} y_{1} y_1^{\dagger} & x_{1} x_2^{\dagger} y_1^{\dagger} y_{2} & x_{2} x_2^{\dagger} y_1^{\dagger} y_{2}\\x_{1} x_1^{\dagger} y_{1} y_2^{\dagger} & x_1^{\dagger} x_{2} y_{1} y_2^{\dagger} & x_{1} x_1^{\dagger} y_{2} y_2^{\dagger} & x_1^{\dagger} x_{2} y_{2} y_2^{\dagger}\\x_{1} x_2^{\dagger} y_{1} y_2^{\dagger} & x_{2} x_2^{\dagger} y_{1} y_2^{\dagger} & x_{1} x_2^{\dagger} y_{2} y_2^{\dagger} & x_{2} x_2^{\dagger} y_{2} y_2^{\dagger}\end{matrix}\right]\right)
\end{align}
satisfies the properties:
\begin{align}
\nonumber
L(x_{1} x_1^{\dagger} y_{1} y_1^{\dagger})+L(x_{2} x_2^{\dagger} y_{1} y_1^{\dagger})&=L((x_{1} x_1^{\dagger}+x_{2} x_2^{\dagger})y_{1} y_1^{\dagger})=L_y(y_{1} y_1^{\dagger}),\\
\nonumber
L(x_{1} x_1^{\dagger} y_{2} y_2^{\dagger})+L(x_{2} x_2^{\dagger} y_{2} y_2^{\dagger})&=L((x_{1} x_1^{\dagger}+x_{2} x_2^{\dagger})y_{2} y_2^{\dagger})=L_y(y_{2} y_2^{\dagger}),\\
\label{eq:PTrace}
L(x_{1} x_1^{\dagger} y_{1} y_2^{\dagger})+L(x_{2} x_2^{\dagger} y_{1} y_2^{\dagger})&=L((x_{1} x_1^{\dagger}+x_{2} x_2^{\dagger})y_{1} y_2^{\dagger})=L_y(y_{1} y_2^{\dagger}).
\end{align}
Hence, by summing up some of the components of the matrix \eqref{eq:matL}, we can recover  the marginal linear operator
$$
M_y=L_y\left(\left[\begin{matrix}y_{1} y_1^{\dagger} & y_{1} y_2^{\dagger} \\
 y_{2} y_1^{\dagger} & y_{2} y_2^{\dagger}\end{matrix}\right]\right)=\left[\begin{matrix}\frac{1}{2} & 0\\ 
0 & \frac{1}{2}\end{matrix}\right],
$$
where the last equality holds when $L((x\otimes y)(x\otimes y)^{\dagger})=\rho_e$, i.e., $M_y=\frac{1}{2}I_2$ is the reduced density matrix of $\rho_e$ on system $B$. The operation we have just described, when applied to a density matrix, is known in QT as \emph{partial trace}. Given the interpretation of $\rho$ as a  dual operator, the 
operation of partial trace simply  follows by Equation \eqref{eq:PTrace}.

Similarly, by partial trace we can obtain  
$$
M_x=L_x\left(\left[\begin{matrix}x_{1} x_1^{\dagger} & x_{1} x_2^{\dagger} \\
 x_{2} x_1^{\dagger} & x_{2} x_2^{\dagger}\end{matrix}\right]\right)=\left[\begin{matrix}\frac{1}{2} & 0\\ 
0 & \frac{1}{2}\end{matrix}\right].
$$
Matrix $M_x$ (analogously to $M_y$) is compatible with probability: there are marginal probabilities whose $M_x$ is the moment matrix, an example being
$$
\frac{1}{2}\delta_{\left\{\begin{bmatrix}
            1\\
            0
           \end{bmatrix}\right\}}(x)+\frac{1}{2}\delta_{\left\{\begin{bmatrix}
            0\\
            1
           \end{bmatrix}\right\}}(x).
$$ 
In other words, we are brought to believe that \emph{marginally} the physical properties $x,y$ of the two particles
have a meaning, i.e., they can be explained through probabilistic mixtures of classical evaluation functionals. We can now ask Nature, by means of a real experiment, to decide between our common-sense notions of how the world works, and Alice's one.
Experimental verification of this phenomenon can be obtained by a CHSH-like experiment, which aims at experimentally reproducing
a situation where~\eqref{eq:condepr0} holds, as explained in Box~2. In this interpretation, the CHSH experiment is an \textit{entanglement witness}, we 
discuss the connection between \eqref{eq:condepr0} and the entanglement witness theorem  in Section \ref{sec:witness}.




{

\begin{framed}

\vspace{-0.5ex}

\noindent {{\bf Box 2: CHSH experiment}}

\vspace{-0.9ex}

{\footnotesize

\noindent 
\begin{center}
\includegraphics[width=12cm]{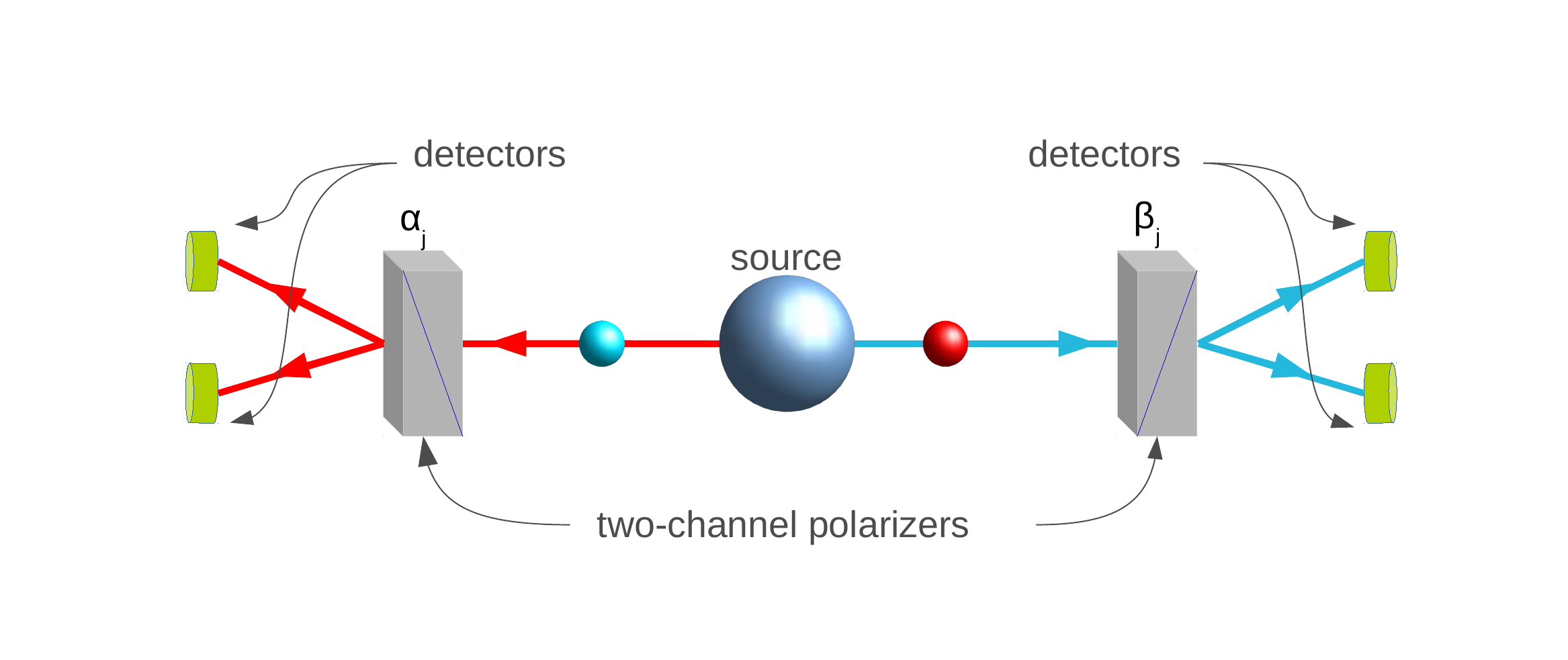}
\end{center}
The source produces pairs of entangled photons, sent in opposite directions. 
Each photon encounters a two-channel polariser whose orientations  can be set by the experimenter. 
Emerging signals from each channel are captured by detectors. 
Four possible orientations $\alpha_i,\beta_j$ for $i,j=1,2$  of the  polarisers are tested.  Consider the Hermitian matrices
$G_{\alpha_i}=\sin(\alpha_i)\sigma_{x}+\cos(\alpha_i)\sigma_{z}$, $G_{\beta_j}=\sin(\beta_j)\sigma_{x}+\cos(\beta_j)\sigma_{z}$,
where $\sigma_{x},\sigma_{z}$ are the 2D Pauli's matrices, and define the gamble 
$$
(x^\dagger G_{\alpha_i} x)(y^\dagger G_{\beta_j} y )=(x \otimes y)^\dagger G_{\alpha_i\beta_j} (x \otimes y)
$$
on the result of experiment with $G_{\alpha_i\beta_j}=G_{\alpha_i} \otimes G_{\beta_j}$. Consider then the sum gamble
$$
h(x,y)=(x \otimes y)^\dagger (G_{\alpha_1\beta_2}-G_{\alpha_1\beta_2}+G_{\alpha_2\beta_1}+G_{\alpha_2\beta_2})(x \otimes y)
$$
and observe that
$$
\begin{aligned}
h(x,y)&=(x \otimes y)^\dagger (G_{\alpha_1}\otimes (G_{\beta_2}-G_{\beta_2}))(x \otimes y) +(x \otimes y)^\dagger (G_{\alpha_2}\otimes (G_{\beta_1}+G_{\beta_2}))(x \otimes y)\\
&=(x^\dagger G_{\alpha_2} x)(y^\dagger (G_{\beta_1}+G_{\beta_2}) y)\\
&\leq y^\dagger (G_{\beta_1}+G_{\beta_2}) y\\
&\leq 2,\\
\end{aligned}
$$
this is the CHSH inequality.  For a small $\epsilon>0$, we have that 
$$
h(x,y)-2-\epsilon=(x \otimes y)^\dagger (-(2+\epsilon) I_4+G_{\alpha_1\beta_2}-G_{\alpha_1\beta_2}+G_{\alpha_2\beta_1}+G_{\alpha_2\beta_2})(x \otimes y)<0
$$ 
\textbf{but} 
$$
~~~~~~~~~~~~~~Tr((-(2+\epsilon)I_4+G_{\alpha_1\beta_2}-G_{\alpha_1\beta_2}+G_{\alpha_2\beta_1}+G_{\alpha_2\beta_2}) \rho_e)=-2+\epsilon+2\sqrt{2}\geq0
$$
for $\alpha_1=\pi/2,\alpha_2=0,\beta_1=\pi/4,\beta_2=-\pi/4$. We are again in a situation like \eqref{eq:condepr0}.
The experiment in the figure certifies the entanglement by measuring the QT expectation of the four components of $h(x,y)$.
 }

\vspace{0ex}

\end{framed}
}

The situation we have just described is the playground of Bell's theorem, stating the impossibility of Einstein's principle of local realism: that is, the combination of the assumption that faraway events cannot influence each other faster than the speed of light (\emph{locality}), with the assumption that properties of objects have a definite, real, value even if they are not measured (\emph{realism}).

The argument goes as follows. 
If we assume that the physical properties $x,y$ of the two particles (the polarization of the photons) have definite values
$\tilde{x},\tilde{y}$ that exist independently of observation, then the measurement on the first qubit
must influence the result of the measurement on the second qubit. 
 Vice versa if we assume \emph{locality}, then $x,y$ cannot exist independently of the observations.
To sum up, a local hidden variable theory that is compatible with QT results cannot exist \citep{bell1964einstein}.

But is there really anything contradictory here? The message we want to convey is that this is not the case.
Indeed, since Theorem~\ref{th:fundamental} applies,  $\rho_e$ is not a moment matrix of any probability.
Ergo, although they may seem to be compatible
with probabilities, the marginal  matrices $M_x,M_y$ are not  moment matrices of any probability. 
The conceptual mistake in the situation we are considering is to forget that $M_x,M_y$ come from $\rho_e$. A joint linear operator uniquely defines its marginals but not the other way round. 
There are infinitely many joint probability charges whose $M_x,M_y$  are the marginals, e.g.,
\[
d\mu(x,y)=\left(\frac{1}{2}\delta_{\left\{\begin{bmatrix}
            1\\
            0
           \end{bmatrix}\right\}}(x)+\frac{1}{2}\delta_{\left\{\begin{bmatrix}
            0\\
            1
           \end{bmatrix}\right\}}(x)\right)\left(\frac{1}{2}\delta_{\left\{\begin{bmatrix}
            1\\
            0
           \end{bmatrix}\right\}}(y)+\frac{1}{2}\delta_{\left\{\begin{bmatrix}
            0\\
            1
           \end{bmatrix}\right\}}(y)\right),
\]
but none of them satisfy  Equation~\eqref{eq:momepr}. Instead, the reader can verify that the  charge in Equation~\eqref{eq:momsigned} satisfies
both Equation~\eqref{eq:momepr} and:

\[
\displaystyle{
\int \begin{bmatrix}x_{1} x_1^{\dagger} & x_{1} x_2^{\dagger} \\
 x_{2} x_1^{\dagger} & x_{2} x_2^{\dagger}\end{bmatrix} d\mu(x,y)=\displaystyle{\int \begin{bmatrix}y_{1} y_1^{\dagger} & y_{1} y_2^{\dagger} \\
 y_{2} y_1^{\dagger} & y_{2} y_2^{\dagger}\end{bmatrix} d\mu(x,y)=\left[\begin{matrix}\frac{1}{2} & 0\\ 
0 & \frac{1}{2}\end{matrix}\right].
}
}
\]

The take-away message of this subsection is that we should only interpret $M_x,M_y$  as marginal operators  and keep in mind that QT is a logical theory of P-coherence.
We see paradoxes when we try to force a physical interpretation upon QT, whose nature is instead computational.
In other words, if we accept that computation is more primitive than our classical interpretation of physics, all paradoxes disappear.



\subsection{Entanglement witness theorem}
\label{sec:witness}

In the previous Subsections, we have seen that all paradoxes of QT emerge because of disagreement between its internal coherence and the attempt to force on it a classical coherent interpretation.

Do quantum and classical probability sometimes agree?
Yes they do, but when at play there are density matrices $\rho$ such that Equation \eqref{eq:condepr0} does not hold, and thus in particular for \emph{separable density matrices}.
We make this claim precise by providing a link between Equation \eqref{eq:condepr0} and the \emph{entanglement witness theorem}
\citep{horodecki1999reduction,horodecki2009quantum}.

We first report the definition of entanglement witness \citep[Sec.~6.3.1]{heinosaari2011mathematical}:
\begin{definition}[Entanglement witness]
\label{def:entangle}
 A Hermitian operator $W \in \He^{n_1 \times n_2} $ is an \emph{entanglement
witness} if and only if $W$ is not a positive operator but $(x_1 \otimes x_2)^{\dagger} W (x_1 \otimes x_2) \geq 0$ for all 
vectors $(x_1,x_2)\in \pspace_1\times \pspace_2$.\footnote{In  \citep[Sec.~6.3.1]{heinosaari2011mathematical}, the last part of this definition says
``for all factorized vectors $x_1 \otimes x_2$''. This is equivalent to considering the pair $(x_1,x_2)$.} 
\end{definition}

The next well-known result (see, e.g.,  \citep[Theorem~6.39, Corollary~6.40]{heinosaari2011mathematical}) provides a characterisation of entanglement 
and separable states in terms of entanglement
witness.
\begin{proposition}
\label{prop:witn}
 A state $\rho_e$ is entangled if and only if there exists an entanglement
witness $W$ such that $Tr(\rho_e W ) < 0$. A state is separable if and only if $Tr(\rho_e W ) \geq 0$ for all entanglement
witnesses $W$.
\end{proposition}


Assume that  $W$ is an entanglement witness for  $\rho_e$ and consider $W'=-W$. By Definition \ref{def:entangle} and Proposition \ref{prop:witn}, it follows
that 
\begin{equation}
\label{eq:witn}
Tr(\rho_e W' ) > 0 \text{ and } (x_1 \otimes x_2)^{\dagger} W' (x_1 \otimes x_2) \leq 0. 
\end{equation}

The first inequality states that the  gamble $(x_1 \otimes x_2)^{\dagger} W' (x_1 \otimes x_2)$ is \emph{strictly} desirable for Alice (in theory $\btheory$) given her belief $\rho_e$.
Since the set of desirable gambles (\ref{eq:b1}) associated to $\rho_e$ is closed, there exists $\epsilon>0$ such that $W''=W'-\epsilon I$ is still desirable, i.e,
$Tr(\rho_e W'' )\geq0$ and 
$$
(x_1 \otimes x_2)^{\dagger} W'' (x_1 \otimes x_2) = (x_1 \otimes x_2)^{\dagger} W' (x_1 \otimes x_2) - \epsilon<0,
$$
where we have exploited that $(x_1 \otimes x_2)^{\dagger} \epsilon I (x_1 \otimes x_2)= \epsilon$.
Therefore, \eqref{eq:witn} is equivalent to
\begin{equation}
\label{eq:witn1}
Tr(\rho_e W'' ) \geq 0 \text{ and } (x_1 \otimes x_2)^{\dagger} W'' (x_1 \otimes x_2) <0,
\end{equation}
which is the same as \eqref{eq:condepr0}.

Hence, by Theorem~\ref{th:fundamental}, we can equivalently formulate the entanglement witness theorem as an arbitrage/Dutch book:
\begin{theorem}
 Let $\bdomain=\{g(x_1,\dots,x_m)=(\otimes_{j=1}^m x_j)^\dagger G (\otimes_{j=1}^m x_j)\mid Tr(G\tilde{\rho})\geq 0\}$ be the set of desirable gambles corresponding to some density matrix $\tilde{\rho}$. The following claims are equivalent:
 \begin{enumerate}
\item $\tilde{\rho}$ is entangled;
\item $\posi(\bdomain \cup \gambles^{\geq})$ is not coherent in $\theory$.
\end{enumerate}
\end{theorem}

This result provides another view of the entanglement witness theorem in light of P-coherence. In particular, it tells us that the existence of a witness satisfying Equation~\eqref{eq:witn} boils down to the disagreement between the classical probabilistic interpretation and the theory $\btheory$ on the rationality (coherence) of Alice, and therefore that whenever they agree on her rationality it means that $\rho_e$ is separable. 

This connection explains why the problem of characterising entanglement  is hard in QT: it amounts to proving the negativity of a function, which is NP-hard.


\section{Entangled states do not only exist in QT}\label{sec:ent_not_only}
In this Section we are going to present an example of entanglement in a P-coherence theory of probability that is different from QT.

Consider two  real variables $x=[x_1,x_2]^T$ with $x_i \in \reals$,  the possibility space $\Omega=\{x\in\reals^2\}$ and the vector space
of gambles
$$
\gambles_R=\{v^T(x)Gv(x): G \text{ is a symmetric real matrix}\},
$$
where $v(x)$ is the column vector of monomials:
$$
v(x)=[1,x_1,x_2,x_1^2,x_1 x_2,x_2^2, x_1^3,x_1^2x_2,x_1x_2^2,x_2^3]^T,
$$
whose dimension is $d=10$. $\gambles_R$ is therefore the space of all polynomials of degree $6$ of the real variables $x_1,x_2$.
It can be observed that in $\gambles_R$ the constant functions can be represented as
$$
v^T(x)Gv(x)=v^T(x)(ce_1e_1^T)v(x)=c,
$$
for any constant $c\in \reals$, $e_1$ being the first element of the canonical basis of $\reals^{d}$, i.e., $e_1=[1,0,\dots,0]$.
Compare this with \eqref{eq:constntQM}.

We say that $\bdomain$  is a P-coherent set of desirable gambles whenever it satisfies \eqref{eq:btaut}--\eqref{eq:b2} in Section \ref{sec:polycoher} with 
\begin{align}
 \Sigma^{\geq}&=\{g(x_1,x_2)=v^T(x)Gv(x): G\geq0\},\\
 \Sigma^{<}&=\{g(x_1,x_2)=v^T(x)Gv(x): G<0\}.
\end{align}
The polynomial $v^T(x)Gv(x)$ with $G\geq0$ are called \textit{sum-of-square polynomials}, see also Section \ref{sec:sos}.

In this case too we can define the dual operator
\begin{equation}
\label{eq:linearoperatorx1x2}
Z:=L_B\left(v(x)v^T(x)\right).
\end{equation}
If we define $z_{\alpha\beta}=L(x_1^\alpha x_2^\beta)$, then $Z$  has the following structure
\begin{align}
\label{eq:dualMatZ}
Z=\begin{bsmallmatrix}
 z_{00} &z_{10} &z_{01} &z_{20} &z_{11} &z_{02} &z_{30} &z_{21} &z_{12} &z_{03}\\
z_{10} &z_{20} &z_{11} &z_{30} &z_{21} &z_{12} &z_{40} &z_{31} &z_{22} &z_{13}\\
z_{01} &z_{11} &z_{02} &z_{21} &z_{12} &z_{03} &z_{31} &z_{22} &z_{13} &z_{04}\\
z_{20} &z_{30} &z_{21} &z_{40} &z_{31} &z_{22} &z_{50} &z_{41} &z_{32} &z_{23}\\
z_{11} &z_{21} &z_{12} &z_{31} &z_{22} &z_{13} &z_{41} &z_{32} &z_{23} &z_{14}\\
z_{02} &z_{12} &z_{03} &z_{22} &z_{13} &z_{04} &z_{32} &z_{23} &z_{14} &z_{05}\\
z_{30} &z_{40} &z_{31} &z_{50} &z_{41} &z_{32} &z_{60} &z_{51} &z_{42} &z_{33}\\
z_{21} &z_{31} &z_{22} &z_{41} &z_{32} &z_{23} &z_{51} &z_{42} &z_{33} &z_{24}\\
z_{12} &z_{22} &z_{13} &z_{32} &z_{23} &z_{14} &z_{42} &z_{33} &z_{24} &z_{15}\\
z_{03} &z_{13} &z_{04} &z_{23} &z_{14} &z_{05} &z_{33} &z_{24} &z_{15} &z_{06}\\
\end{bsmallmatrix}
\end{align}
It is not difficult to verify that, given $X=v(x)v^T(x)$, in this case Lemma \ref{lem:uniqueX} does not hold. This is actually the reason why the dual matrix $Z$ has additional
symmetries. 

Now, the dual of $\bdomain$ can be calculated as in Section \ref{sec:dualQM}:
\begin{equation}
\label{eq:dualM1sos}
\begin{aligned}
\mathcal{Q}&=\left\{{z} \in \reals^{{d}}:  L_B(g)=Tr({G} Z)\geq0, ~Z\geq 0,~ z_{00}=1, ~~\forall g \in \bdomain\right\}.
\end{aligned}
\end{equation} 
Note that, because the way constants are represented in this context, we have $z_{00}=1$ instead of $Tr(Z)=1$.

Since sum-of-squares implies nonnegativity, a natural question is to known whether any nonnegative polynomial
of degree $6$ of the real variables $x_1,x_2$ can be expressed
as a sum of squares. 
It turns out that this is not the case. Indeed, David Hilbert showed that equality between the
set of nonnegative polynomials of $n$ variables of degree $2d$ and sos polynomials of $n$ variables of degree $2d$ holds only in the
following three cases: univariate polynomials (i.e., $n = 1$); quadratic polynomials ($2d = 2$); bivariate quartics ($n = 2$, $2d = 4$).
For all other cases, there always exist nonnegative polynomials that are not sums
of squares. 
The most famous counter-example is a polynomial due to Motzkin:
$$
m(x)=x_1^4x_2^2+x_1^2x_2^4-x_1^2x_2^2+1.
$$

Motzkin polynomial is nonnegative but it is not a sum-of-squares.
In what follow, it will be used to prove that a theory $\theory^*$ in the space of polynomials
of two real variables of degree $6$  has ``entangled density matrices'' in its dual space.

Consider the following PSD matrix \cite{Benavoli2017b} of type \eqref{eq:dualMatZ}
\begin{align}
\label{eq:darioMat}
Z_e=\begin{footnotesize}\begin{bsmallmatrix} 1 & 0 & 0 & 353 & 0 & 353 & 0 & 0 & 0 & 0\\ 0 & 353 & 0 & 0 & 0 & 0 & 249572 & 0 & 66 & 0\\ 0 & 0 & 353 & 0 & 0 & 0 & 0 & 66 & 0 & 249572\\ 353 & 0 & 0 & 249572 & 0 & 66 & 0 & 0 & 0 & 0\\ 0 & 0 & 0 & 0 & 66 & 0 & 0 & 0 & 0 & 0\\ 353 & 0 & 0 & 66 & 0 & 249572 & 0 & 0 & 0 & 0\\ 0 & 249572 & 0 & 0 & 0 & 0 & 706955894 & 0 & 17 & 0\\ 0 & 0 & 66 & 0 & 0 & 0 & 0 & 17 & 0 & 17\\ 0 & 66 & 0 & 0 & 0 & 0 & 17 & 0 & 17 & 0\\ 0 & 0 & 249572 & 0 & 0 & 0 & 0 & 17 & 0 & 706955894 \end{bsmallmatrix}\end{footnotesize}
\end{align}
and the negative  polynomial $-m(x)$.
Since  $L(f)=-1-z_{42}-z_{24}+z_{22}$ and   $z_{22}=66, z_{24}=z_{42}=17$ in \eqref{eq:darioMat},  we have that $L(f)=31>0$.
Therefore,  the conditions in \eqref{eq:condepr0} are met and Theorem \ref{th:fundamental} holds.
In particular, this means that the polynomial $-m(x)$ is desirable by a subject, Alice, whose beliefs are expressed by $Z_e$ and who is therefore a rational agent in 
$\theory^*$. However,  the polynomial is negative: Alice is irrational in 
$\theory$.

Notice that, if we consider the definition of entanglement in Section \ref{sec:witness},  
$Z_e$ is an entanglement matrix.
Moreover, from \eqref{eq:darioMat} and the definition of $Z$ we can extract the marginal operator for the variables $x_1,x_2$.
\begin{align}
 M_{x_1} &=L\left(\begin{bsmallmatrix}
                1 & x_1 & x_1^2 & x_1^3\\
                x_1 & x_1^2 & x_1^3 & x_1^4\\
                x_1^2 & x_1^3 & x_1^4 & x_1^5\\
                x_1^3 & x_1^4 & x_1^5  & x_1^6\\
               \end{bsmallmatrix}
\right)=
\begin{bsmallmatrix}
 1 & 0 & 353 & 0\\
 0 & 353 & 0& 249572\\
 353 & 0 & 249572 & 0\\
 0 & 249572 & 0 & 706955894
\end{bsmallmatrix},\\
 M_{x_2}&=L\left(\begin{bsmallmatrix}
                1 & x_2 & x_2^2 & x_2^3\\
                x_2 & x_2^2 & x_2^3 & x_2^4\\
                x_2^2 & x_2^3 & x_2^4 & x_2^5\\
                x_2^3 & x_2^4 & x_2^5  & x_2^6\\
               \end{bsmallmatrix}
\right)=
\begin{bsmallmatrix}
 1 & 0 & 353 & 0\\
 0 & 353 & 0& 249572\\
 353 & 0 & 249572\\ & 0\\
 0 & 249572 & 0 & 706955894
\end{bsmallmatrix}.
\end{align}
This time the operation is not equal to partial trace because, again, the structure of the polynomials $v^T(x)Gv(x)$ 
is different from those in QT.  

Since $ M_{x_1}, M_{x_2}$ are both PSD they are valid moment matrices.
Moreover, for single real variables, Hilbert as shown that every nonnegative polynomial is SOS and, therefore,
marginally, the above truncated moment matrices are compatible with classical probability. A subject whose marginal beliefs are  expressed by $ M_{x_1}\geq0, M_{x_2}\geq0$ is always rational in $\theory$.
This is exactly the same situation encountered when showing the impossibility of a local realistic interpretation of quantum mechanics (see Section \ref{sec:local}): the marginals seem compatible with  classical probability but the joint is not, since we can derive a Bell-type inequalities that are violated by  $Z_e$.

Again the conceptual mistake of this reading  is  to forget that $M_x$ and $M_y$ come from $Z_e$ and thus they should only be interpreted  as marginal operators .


\section{Discussions}
\label{sec:discussions}


\subsection{First postulate of QT from polynomial time complexity}
\label{sec:discussions0}
 The first postulate of QT is  formulated as follows:
 \begin{quote}
  Associated to any isolated physical system is a complex vector space
with inner product (that is, a Hilbert space) known as the state space of the
system. The system is completely described by its density operator, which is a
positive operator $\rho$ with trace one, acting on the state space of the system.
 \end{quote}
 This postulate is usually explained by noticing that the density matrix $\rho$ is the quantum-mechanical analogue to a phase-space probability measure  in classical statistical mechanics.  One of the fundamental questions in QT is why quantum experiments cannot be described via a classical phase-space probability measure.

In this paper, we argued that a possible explanation may be found in the computational resources necessary to address consistency (rationality)
in classical probability. 
Addressing rationality in the classical case reduces to prove that a real-valued bounded function is nonnegative.
However, the problem of verifying the nonnegativity of a function is undecidable or, when decidable,  is in general NP-hard.

If the universe is ``inherently'' probabilistic and if QT is the theory of physics that describes such probabilistic nature,
QT should be first a computable theory of probability  and, moreover, computationally efficient. 
A way to define a computationally efficient theory of probability is (i) to restrict the space
of functions of interest (observables); (2) to redefine the meaning of  nonnegative/negative.
 From the point of view of a logical theory, these two points mean, firstly to set the appropriate language (the syntax of the theory), and secondly, to define which formulas are the tautologies and which are the contradictions, the remaining apparatus of the theory being structurally the same as before.

The key postulate that separates classical probability and QT is~\ref{eq:btaut}: the computation postulate. Because of~\ref{eq:btaut}, Theorem~\ref{th:fundamental} applies and thus the ``weirdness'' of QT follows: negative probabilities, existence of non-classical evaluation functionals and, therefore, irreconcilability with the classical probabilistic view.


The formulation of Theorem \ref{th:fundamental} points to the fact that there are
three equivalent ways to provide a theoretical foundation of QT.

The \underline{first way} is the one adopted within the \textit{ Quantum Logic} (QL) approach. QL justifies the differences between QT and classical probability  with the premise that, in QT, the Boolean algebra of events is taken over by the ``quantum logic'' of projection operators on a Hilbert space.
This gives raise to a new theory of probability. The QL probabilistic formalism was developed by \cite{birkhoff1936logic} and then formalised by   Mackey's 
eight axioms \cite{mackey2013mathematical}. The key axiom (Axiom VII) defines the algebra of the events and can be seen as a formalisation
of the properties of what we have called ``non classical evaluation functions''.
We want point out that ``non classical evaluation functions'' only exist when looking from the outside of the theory as a way to justify 
how a negative polynomial can assume positive values in a quantum experiment.


The \underline{second way} is the so called \textit{quasi-probability} (QP) formulations of QT. It is based on the view that the quantum-classical contrast in QT is due to the appearance of negativity. 
In QP, the possibility space (and  the events) is classical but probability distributions are replaced
by quasi-probability distributions. 

One of the critical issues with negative probabilities is their interpretation. 
This is clearly stated in \cite{ferrie2011quasi}
\begin{quote}
 The main difficulty with interpreting negativity in a particular quasi-probability representation as a criterion
for or definition of quantumness is the non-uniqueness of
that particular quasi-probability representation. We can
always find a new representation in which any given state
admits a non-negative quasi-probability representation.
Recall, in fact, that in some representations all states
are non-negative. Thus, negativity of some state $\rho$ in
one particular arbitrary representation is a meaningless
notion of quantumness per se.
\end{quote}
We agree that quasi-probability distributions are a meaningless notion per se.\footnote{Notice that in the theory $\theory^*$, quasi-probability distributions do not exist, the dual space being $\{\rho: \rho\geq0,~~Tr(\rho)=1\}$. }
In our view, their only justification is thorough a notion of coherence (logical consistency), as done with an abstract notion of classical probability theory.

Why is the charge in Box~1 a valid quasi-probability representation?
The reason is that it corresponds to the dual of a maximal P-coherent set of  desirable gambles $\bdomain$.
 A subject (Alice) that uses that  charge to make rational choices on the results of a quantum experiment cannot be made
 a sure loser in the theory $\theory^*$.
 The fact that charges are not unique is neither surprising nor a characteristic of QT.
 In classical probability, it is well known that any truncated moment matrix defines
 a closed convex set of probability distributions.
Similarly, one can check that any ``quasi-moment'' matrix $\rho$ defines
  a closed affine set of quasi-probability distributions.
Such intrinsic non-uniqueness is another amazing characteristic of QT: from the perspective of a probabilistic interpretation, 
  QT is not only a  theory of probability, but a theory
  of imprecise probability.

The \underline{third way}, championed in this paper, argues that the quantum-classical contrast has a purely computational character.
It starts by accepting the so-called  Exponential Time Hypothesis (P$\neq$NP) to justify the
separation between the microscopic quantum system and the macroscopic world.
We quote Scott Aaronson \cite{aaronson2005guest}
\begin{quote}
 can NP-complete problems
be solved in polynomial time using the resources of the physical universe?
I will argue that studying this question can yield new insights, not just about computer science
but about physics as well. More controversially, I will also argue that a negative answer might
eventually attain the same status as (say) the Second Law of Thermodynamics, or the impossibility
of superluminal signalling. In other words, while experiment will always be the last appeal, the
presumed intractability of NP-complete problems might be taken as a useful constraint in the search
for new physical theories.
\end{quote}
\ref{eq:btaut}  may indeed be the fundamental law in QT, similarly to the   Second Law of Thermodynamics, or the impossibility
of superluminal signalling. 
To the best knowledge of the authors, the present work is the  first that explains QT in terms of computational complexity.


\subsection{The class of P-nonnegative gambles}
\label{sec:HermSOS}
The class of P-nonnegative gambles, defined in Section \ref{sec:polycoher}, is the closed convex cone of all \textit{Hermitian sum-of-squares} in $\gambles_R=\{(\otimes_{j=1}^m x_j)^\dagger G (\otimes_{j=1}^m x_j) \mid G\in \He^{r \times r}\}$, that is of all gambles $ g(x_1,\dots,x_m) \in \gambles_R$ for which $G$ is PSD. In particular this means that 
 Alice can efficiently determine whether a gamble is P-nonnegative or not.  
But is this class the only closed convex cone of nonnegative polynomials in $\gambles_R$ for which the membership problem can be solved efficiently (in polynomial-time)?
It turns out that the answer is negative (see for instance \cite{d2009polynomial,josz2018lasserre}): in addition of \textit{Hermitian sum-of-squares} (the one that Nature has chosen for QT)
one could also consider \textit{real sum-of-squares} in $\gambles_R$, that is polynomials of the form $(\otimes_{j=1}^m x_j)^\dagger G (\otimes_{j=1}^m x_j)$ that are sum-of-squares of polynomials of the real and imaginary part of the variables $ x_j$.

A separating example is the polynomial in \eqref{eq:sosqt}, which is not a \textit{Hermitian sum-of-squares} but it is a \textit{real sum-of-square}, as it can be seen from  \eqref{eq:realim}. This polynomial was used in our example because it can be constructed by inspection and its nonnegativity
follows immediately by \eqref{eq:realim}.  Clearly, there exist nonnegative polynomials in $\gambles_R$ that are neither Hermitian sum-of-squares nor real sum-of-squares.

Why has Nature chosen Hermitian sum-of-squares? 
This is an open question that we will investigate in future work. A possible explanation may reside
in the different size of the corresponding optimisation problems \cite{josz2018lasserre}.

%
%
%
%


\subsection{A hidden variable model for a single quantum system}
\label{sec:hidden1d}
In \cite{kochen1968problem},  Kochen and Specker gave a  hidden variable model agreeing with the Born's rule but not preserving the structure of functional dependencies in QT.
Their idea amounted to introducing a ``hidden variable'' for each observable $H$ producing stochasticity in outcomes
of measurement of $H$. The totality of all such hidden variables is then
the phase space variable $\omega$ of the model.

It turns out that, for a single quantum system, our  model based on the phase space
$$
\Omega=\{x\in \complex^n: x^{\dagger}x=1\},
$$
is also a hidden variable model. The reason being that, for a single quantum system, $\Sigma^{\geq}=\gambles_R^{\geq}$
and $\Sigma^{<}=\gambles_R^{<}$, and therefore Alice will never accepts negative gambles. Notice that in this case, the matrix $\rho=L(xx^{\dagger})$ can be interpreted as a moment matrix as discussed in Section \ref{sec:momentmat}.
This hidden-variable model for QT is also discussed in \cite[Sec.~1.7]{holevo2011probabilistic}, where the author  explains
why this model does not contradict the existing ``no-go'' theorems for hidden-variables.
Our hidden variable model differs from Holevo's one when we consider $m>1$ particles.
This is discussed in the next section.


\subsection{On the use of tensor product}
\label{sec:tensorproduct}

We have seen in Section \ref{sub:space} that the possibility space of composite systems of $m$ particles, each one with $n_j$ degrees of freedom, is given by 
$\pspace=\prod_{j=1}^m \overline{\complex}^{n_j}$, and that gambles (real-valued observable) on such space  are actually bounded real functions $g(x_1,\dots,x_m)=(\otimes_{j=1}^m x_j)^\dagger G (\otimes_{j=1}^m x_j)$,  where $\otimes$ denotes the tensor product between vectors, seen as column matrices.

In what follows, we justify the use of tensor product, and more specifically the type of gambles on the possibility space of composite systems, as a consequence of the way a multivariate theory of probability is usually formulated. 

As a start, let us consider the case of 
classical probability. There,  structural judgements of independence/dependence between random variables are expressed thorough product
between random variables. Indeed, given  factorised gambles $g(x_1,\dots,x_m)=\prod_{j=1}^m g_j(x_j)$, the variables $x_1,\dots,x_m$ are said to be independent
if $E[\prod_{j=1}^m g_j(x_j)]=\prod_{j=1}^m E[g_j(x_j)]$ for all $g_j$, where $E[\cdot]$ denote the expectation operator.

With this in mind, let us go back to our setting. Marginal gambles are of type $g_j(x_j)=x_j^{\dagger}G_j x_j$. This means that structural judgements are performed by considering factorised gambles of the form $\prod_{j=1}^m x_j^{\dagger}G_j x_j$. As mention in Section \ref{sub:space}, it is thence not difficult to verify that
$$
\prod_{j=1}^m x_j^{\dagger}G_j x_j=(\otimes_{j=1}^m x_j)^\dagger (\otimes_{j=1}^m G_j) (\otimes_{j=1}^m x_j).
$$
By closing the set of factorised gambles  under the operations of addition and scalar (real number) multiplication, 
one finally gets a vector space whose domain domain coincides with the collection of all gambles of the form $(\otimes_{j=1}^m x_j)^\dagger G (\otimes_{j=1}^m x_j)$.
Hence, structural judgements of independence/dependence are stated by Alice considering the desirability
of gambles belonging to $\gambles_R$.


In QT, there is   a confusion over  the  role of the tensor product used to define the state space.
This confusion is a consequence of Theorem~\ref{th:fundamental}.
By duality, we can prove that, given two quantum systems $A$ and $B$, with corresponding Hilbert spaces $\mathcal{H}_A$
and $\mathcal{H}_B$, the density matrix (state) of the joint system  lives in the tensor product space $\mathcal{H}_A\otimes \mathcal{H}_B$.
This follows by:
$$
\begin{aligned}
L((\otimes_{j=1}^2 x_j)^\dagger G (\otimes_{j=1}^2 x_j))&=L(Tr(G(\otimes_{j=1}^2 x_j)(\otimes_{j=1}^2 x_j)^\dagger))\\
&=Tr(G\; L((\otimes_{j=1}^2 x_j)(\otimes_{j=1}^2 x_j)^\dagger)).
\end{aligned}
$$
and $L((\otimes_{j=1}^2 x_j)(\otimes_{j=1}^2 x_j)^\dagger)$ belongs to $\mathcal{H}_A\otimes \mathcal{H}_B$.

However, when \eqref{eq:condepr0} holds, we may justify  entanglement hypothesising  the existence of non classical evaluation functions
or, equivalently, a larger possibility space (Theorem \ref{th:fundamental}). This is clearly discussed in \cite[Supplement 3.4]{holevo2011probabilistic}:
\begin{quote}
 Since the set of pure states of the composite system $\Omega$ is larger
than Cartesian product $\Omega_A\times \Omega_B$, the phase space of the classical description of the
composite system will be larger than the product of phase spaces for the
components: $\Omega_A\times \Omega_B \varsubsetneq \Omega$. Therefore this classical description is not a
correspondence between the categories of classical and quantum system
preserving the operation of forming the composite systems.
Moreover, it appears that there is no way to establish such a correspondence. In any classical description of a composite quantum system
the variables corresponding to observables of the components are necessarily entangled in the way unusual for classical subsystems.
\end{quote}
We argue that $\Omega_A\times \Omega_B \varsubsetneq \Omega$ is a manifestation of computational rationality.

\subsection{Deriving the remaining postulates of QM}\label{subsec:otherax}

Partially inspired by \cite{pitowsky2003betting}, in \cite{benavoli2016quantum} we introduced a syntactically different framework from which we were able to derive the postulates of QT. It is roughly defined as follows. 
Assume the space of outcomes of experiment on a $n$-dimensional quantum systems is represented by the set $\Omega=\{\omega_1,\dots,\omega_n\}$,  with $\omega_i$ denoting the elementary event ``detection along $i$''. 
A gamble on an experiment 
is a Hermitian matrix $G \in \He^{n\times n}$. 
By accepting  a gamble $G$, Alice commits herself to receive  $\gamma_{i}\in \reals$ utiles  if the outcome of the experiment eventually happens to be 
$\omega_i$, where $\gamma_{i}$ is defined from $G$ and the projection-valued measurement $\Pi^{*}=\{\Pi_i^*\}_{i=1}^n$, representing the $n$ orthogonal directions of the quantum state, as follows:
 \begin{equation}
  \Pi^{*}_{i}G\Pi^{*}_{i}=\gamma_{i}\Pi^{*}_{i} \text{ for } i=1,\dots,n.
 \end{equation} 

A subset $ \He^{n\times n}$ is thus said to be coherent if it is a  closed convex cone $\mathscr{C}$ containing the set  $\{G\in\He^{n\times n}:G\gneq0\}$ of all PSD matrices in $\He^{n\times n}$ and disjoint from the interior of $\{G\in\He^{n\times n}:G \leq 0\}$.
We then proved that the dual of a coherent  convex cone of matrix gambles $\mathscr{C}$  is a closed convex set of density matrices
\cite[Prop.IV.3]{benavoli2016quantum}:
\begin{equation}
\label{eq:dualrhoo}
\begin{aligned}
\mathcal{Q}
&=\left\{ \rho \in  \He^{n\times n}: \rho\geq0, Tr(\rho)=1, Tr({G} \rho)\geq0,  ~~\forall G  \in \mathscr{C}\right\}.
\end{aligned}
\end{equation} 
When  $\mathscr{C}$ is a maximal cone,  its dual $\mathcal{Q}:=\mathscr{C}^{\bullet}$ includes a single density matrix.
This allowed us to show that QT is a generalised theory of probability: its axiomatic foundation can be derived from a logical consistency requirement in the way a subject accepts gambles on the results of a quantum experiment (similar to the axiomatic foundation of classical probability).
However, at the time of writing \cite{benavoli2016quantum}, it was not clear to us why the probability is generalised in such a way in QT, what this could possibly mean, why does entanglement exist, etc.
That is, it was not clear to us that everything follows by the computation postulate \ref{eq:btaut}.

Below, by providing a bijection between coherent  convex cones of matrices and  P-coherent systems of polynomial gambles,
we show the connection between \cite{benavoli2016quantum}  and the present work.
In order to that, we first need to replace \cite[Definition III.1,(S3)]{benavoli2016quantum}, that is the \textit{openness} property, with (S3'): 
if $G+\epsilon I \in \mathscr{C}$ for all $\epsilon>0$ then  $G$ is in $ \mathscr{C}$ (\textit{closeness}). 
Hence, the bijection $\mathfrak{f}$ between $\mathscr{C}$ and $\bdomain$ is obtained  by first noticing the correspondence between  the duals  \eqref{eq:dualM1} and \eqref{eq:dualrhoo}, and thus simply composing
the duality maps from $\bdomain$ to $\mathcal{Q}$ and from $\mathcal{Q}$ to  $\mathscr{C}$: 

{\footnotesize
\[
\begin{tikzcd}[column sep=large, row sep=large]
\mathcal{Q}   \arrow{rd}{} 
  &  \arrow{l}{} \overbrace{\left\{\sum_{i=1}^{|\mathcal{G}|}\lambda_i (\otimes_{j=1}^m x_j)^{\dagger}G_i(\otimes_{j=1}^m x_j)+ (\otimes_{j=1}^m x_j)^{\dagger}M (\otimes_{j=1}^m x_j)~: ~~\lambda_i\geq0, M\geq0 \right\}  }^{\bdomain=posi(\mathcal{G}\cup\Sigma^{\geq})}\arrow{d}{\mathfrak{f}} \\
    &  \underbrace{\left\{\sum_{i=1}^{|\mathcal{G}|}\lambda_i G_i + M~:~ ~\lambda_i\geq0, M\geq0\right\}}_{\mathscr{C}}
\end{tikzcd}
\]}


Based on $\mathfrak{f}$ and the results in \cite{benavoli2016quantum}, it is therefore almost\footnote{We need to take into account that (S3) has been replaced by (S3').} immediate to derive from \ref{eq:btaut}--\ref{eq:b2} the remaining other axioms and rules of QT (such as L\"{u}der's rule (measurement updating), Schr\"odinger's rule (time evolution)), see Appendix \ref{app:otherax}.



\subsection{Sum-of-squares optimisation}
\label{sec:sos}
The theory of moments (and its dual theory of positive polynomials) are used to develop efficient numerical
schemes for polynomial optimization, i.e., global optimization problems with polynomial function. 
Such problems arise  in the analysis and control of nonlinear dynamical systems, and also in other areas such as combinatorial optimization.
This scheme consists of a hierarchy of semidefinite
programs (SDP) of increasing size which define tighter and tighter relaxations of the original
problem. Under some assumptions, it can be showed that the associated sequence of optimal values converges to the global minimum, see for instance \cite{parrilo2003semidefinite,lasserre2009moments}.
Note that, every polynomial in 
$$
\Sigma^{\geq}\coloneqq\{g(x_1,\dots,x_m)=(\otimes_{j=1}^m x_j)^\dagger G (\otimes_{j=1}^m x_j): G\geq0\},
$$
is (Hermitian) sum-of-squares because it can be rewritten as:
$$
(\otimes_{j=1}^m x_j)^\dagger H H^{\dagger} (\otimes_{j=1}^m x_j)=\sum_{i=1}^k |(H^{\dagger} (\otimes_{j=1}^m x_j))_i|^2,
$$
with $G= HH^{\dagger}$.

We have recently discussed the connection between SOS optimisation (for polynomials of real variables) and computational rationality (also called bounded rationality) in \cite{Benavoli2019b}.

In QT, SDP has been used to numerically prove that a certain state is entangled  \cite{landau1988empirical,doherty2002distinguishing,doherty2004complete, wehner2006tsirelson,doherty2008quantum,navascues2008convergent,pironio2010convergent,bamps2015sum,barak2017quantum}. 
The work \cite{doherty2002distinguishing,doherty2004complete}  realized that the set of separable
quantum  states  can  be  approximated  by  sum-of-squares  hierarchies. 
This  leads  to  the  SDP hierarchy  of  Doherty-Parrilo-Spedalieri,  which  is  extensively  employed  in  quantum  information.

The present, purely foundational, work differs from these approaches by stating that the Universe (microscopic world)  is nothing but a big ``device'' that solves SOS
optimisation problems. In fact, as discussed in Section \ref{sec:discussions0},  the postulate of computational efficiency embodied by~\ref{eq:btaut} (through the above definition of $\Sigma^{\geq}$, i.e., the cone of Hermitian sum-of-squares polynomials) may indeed be the fundamental law in QT.
This totally different perspective may have an enormous impact in quantum computer, such as the development of new algorithms for quantum computing that exploit
the connection between QT and SDP highlighted in this paper.

\begin{appendix}
\section{Technicalities for Subsection \ref{subsec:dual}}

Recall that the dual $\domain^\bullet$ of $\domain$  is given by 
$\left\{\mu\in \mathcal{M}: \int gd\mu \geq0, ~\forall g \in \domain \right\}$.
Similarly, the dual of a subset $\mathcal{R}$ of $\mathcal{M}$ is the set:\\
\begin{equation}
\mathcal{R}^\bullet=\left\{g\in \mathcal{L}: \int gd\mu \geq0, ~\forall \mu \in \mathcal{R}\right\}.
\end{equation}

Note that in both cases $(\cdot)^\bullet$ is always a closed convex cone \citep[Lem.5.102(4)]{aliprantisborder}. Furthermore,  one has that $(\cdot)^\bullet{^\bullet}=(\cdot)$, whenever $(\cdot)$ is a closed convex cone  \citep[Th.5.103]{aliprantisborder}, and   $(\cdot)_1 \subseteq (\cdot)_2$ if and only if $(\cdot)_2^\bullet \subseteq (\cdot)_1^\bullet$ \citep[Lem.5.102(1)]{aliprantisborder}. In particular, whenever $(\cdot)_1$ and $(\cdot)_2$ are closed convex cones, $(\cdot)_1 \subsetneq (\cdot)_2$ if and only if $(\cdot)_2^\bullet \subsetneq (\cdot)_1^\bullet$.

Based on those facts, it is thus possible to verify that the dual of a coherent set of desirable gambles can actually be completely described in terms of a (closed convex) set of states (probability charges). 
In this aim, we start by the following observations.

\begin{proposition}\label{prop:czesc}
 It holds that 
 \begin{enumerate}
   \item $(\mathcal{L})^\bullet=\{0\}$ and $\mathcal{L}=(\{0\})^\bullet$;
 \item $(\nonnegative)^\bullet=\mathcal{M}^{\geq}$ and $\nonnegative=(\mathcal{M}^{\geq})^\bullet$;
  \end{enumerate}
\end{proposition}
\begin{proof} Since $(\cdot)^\bullet{^\bullet}=(\cdot)$, whenever $(\cdot)$ is a closed convex cone, in both cases it is enough to verify only one of the claims.
For the first item, the second claim is immediate.
For the second item, we verify the first claim.
The inclusion from right to left being clear, for the other direction  observe that: (i) $ g=I_{\{B\}}$ (with $I_B$ being the indicator function on $B \in \mathcal{A}$), is a nonnegative gamble and, therefore, is in $\nonnegative$; (ii) if $\mu$ is negative in $B \in \mathcal{A}$, i.e.,  then $\int I_{B}d\mu$ is negative too and, thus, $\mu$ cannot be in $(\nonnegative)^\bullet$.  
\end{proof}
\begin{proposition}\label{prop:czesc2}
Let $\domain $ be a closed convex cone. The following claims are equivalent
 \begin{enumerate}
   \item $\domain$ is coherent;
     \item $\domain\supseteq \nonnegative$ and $\domain \neq \mathcal{L}$;
   \item $(\domain)^\bullet\subseteq \mathcal{M}^{\geq}$ and $(\domain)^\bullet \neq \{0\}$.
  \end{enumerate}
\end{proposition}
\begin{proof} (2) $\Leftrightarrow$ (3): From Proposition \ref{prop:czesc}, $(\domain)^\bullet = \{0\}$ if and only if $\domain = \mathcal{L}$, and $\domain\supseteq \nonnegative$ if and only if $(\domain)^\bullet\subseteq \mathcal{M}^{\geq}$.
\\ (1) $\Rightarrow$ (2): Assume $\domain$ is coherent. By A.1 $\domain \supseteq \nonnegative$ and by A.2 there is $g \in \mathcal{L}$ such that $\sup g < 0$ and $g \notin \domain$. 
\\ (2) $\Rightarrow$ (1): Let $\nonnegative \subseteq \domain \subsetneq \mathcal{L}$. First of all, notice that, by Proposition \ref{prop:czesc}, $(\domain)^\bullet\subseteq (\nonnegative)^\bullet = \mathcal{M}^{\geq}$. 
Now, assume that $\domain$ is not coherent. This means A.2 fails, i.e. there is $g \in \mathcal{L}$ such that $\sup g < 0$ and $g \in \domain$. 
Hence, consider $\mu \in \mathcal{M}^{\geq}$ and pick $g \in \domain$ such that $\sup g < 0$. It holds that $\langle g, \mu\rangle \geq 0$ if and only if $\mu=0$, meaning that $(\domain)^\bullet=\{0\}$ and therefore, by Proposition \ref{prop:czesc} again, $\domain=\mathcal{L}$, a contradiction.
\end{proof}



\begin{proofof}{Theorem~\ref{prop:dualcharges0}}
We want to verify that the map
\[\domain \mapsto  \mathcal{P}:=\domain^\bullet \cap \mathscr{S} \]
establishes a bijection between coherent sets of desirable gambles and non-empty closed convex sets of states.
The proof is analogous to that by \citep[Th.4]{pmlr-v62-benavoli17b}.
 Let $\domain$ be a coherent set of desirable gambles. By Proposition \ref{prop:czesc2}, we get that $\domain^\bullet$ is a closed convex cone included in $\mathcal{M}^{\geq}$ that does not reduce to the origin.  
Thus, after normalisation, $\mathcal{P}$ is nonempty. Preservation of closedness and convexity by finite intersections yields that $\mathcal{P}$ is closed and convex. Furthermore $\mathbb{R}_+\mathcal{P}=\domain^\bullet$, and therefore $\domain= (\mathbb{R}_+\mathcal{P})^\bullet$, where $\mathbb{R}_+\mathcal{P}:=\{ \lambda\mu : \lambda \geq 0, \mu \in \mathcal{P}\}$, meaning that the map is an injection.
We finally verify that the map is also a surjection. To do this, let $\mathcal{P}$ be a non empty closed convex set of probability charges. It holds that $\mathbb{R}_+\mathcal{P}$ is a closed convex cone included in $\mathcal{M}^{\geq}$ different from $\{0\}$. Again by Proposition \ref{prop:czesc2},  we  conclude that the dual $(\mathbb{R}_+\mathcal{P})^\bullet$ of $\mathbb{R}_+\mathcal{P}$ is a coherent set of desirable gambles and $\mathcal{P}=\mathbb{R}_+\mathcal{P}\cap \mathscr{S}= (\mathbb{R}_+\mathcal{P}){{}^\bullet{^\bullet}}\cap \mathscr{S}$. 
\end{proofof}


As an immediate corollary of the previous results, we finally obtain Theorem \ref{cor:noncoherent}.
It provides us with a characterisation of the dual of a closed convex cone which is not coherent.
\begin{theorem}\label{cor:noncoherent}
Let $\domain $ be a non empty closed convex cone. Then the following are equivalent
\begin{enumerate}
 \item $\domain  \neq \mathcal{L}$  and $\domain$ is not coherent;
 \item 
 $\domain^\bullet\not\subseteq \mathcal{M}^{\geq}$,
 \item $\domain^\bullet \cap \{ \mu \in \mathcal{M} \mid \langle \mathbbm{1}, \mu \rangle =1 \} \not \subseteq \mathscr{S}$.
  \item $\{0\}\subsetneq\domain^\bullet$ and $\domain^\bullet \cap   \mathscr{S}= \emptyset$.
 \end{enumerate} 
\end{theorem}
Essentially, Theorem \ref{cor:noncoherent} is telling us that, from the dual point of view, non degenerated closed convex cones of gambles that are not coherent are characterised by signed (non positive) charges. 

\begin{proofof}{Theorem~\ref{cor:noncoherent}}
Given  a non empty closed convex cone $\domain $, we need to check  that the following claims are equivalent
\begin{enumerate}
 \item $\domain  \neq \mathcal{L}$  and $\domain$ is not coherent;
 \item 
 $(\domain)^\bullet\not\subseteq \mathcal{M}^{\geq}$,
 \item $\domain^\bullet \cap \{ \mu \in \mathcal{M} \mid \langle 1, \mu \rangle =1 \} \not \subseteq \mathscr{S}$.
  \item $\{0\}\subsetneq\domain^\bullet$ and $\domain^\bullet \cap   \mathscr{S}= \emptyset$.
 \end{enumerate} 
Notice that $\domain^\bullet \supseteq \{0\}$, and that, by Proposition  \ref{prop:czesc}, all claims imply $\domain^\bullet \neq \{0\}$ and $\domain \neq \gambles$. The equivalence between (2) and (3) being obvious, the one between (1) and (2) is an immediate consequence of Proposition \ref{prop:czesc2}.  For the equivalence between (1) and (4), first of all notice that, since $\domain $ is a non empty closed convex cone, 
$\{0\}\subseteq\domain^\bullet$ if and only if $\domain  \neq \mathcal{L}$. Now, assume  $\domain$ is not coherent.  This means there is a negative gamble $g$ in  $\domain$. Hence, $\domain^\bullet \cap   \mathscr{S}= \emptyset$ since probability charges do not preserves negative gambles. Finally, if $\domain$ is coherent, by Theorem \ref{prop:dualcharges0},   $\domain^\bullet \cap   \mathscr{S} $ is a non-empty closed convex cone of probability charges. 
\end{proofof}

  \section{Relative coherence}\label{app:comp}

   Why do we need $\theory$ if it is dual to classical probability theory?
   The point is that, $\theory$, having the structure of an abstract logic, is independent of the specific contingent properties of the underlying space of gambles. From this perspective, it is thence  more general than probability theory and can be used to model any circumstance in which an agent has to make a rational choice.
   
For instance, we can define the desirability postulates on any vector subspace of $\gambles$ that includes
constant gambles. This is useful in applications where the gambles Alice can examine are a subset of $\gambles$.\footnote{
For example, in Finance the tradeable gambles are usually piecewise polynomials.
An European call option on the future value $x$ of an underlying security  with strike $k$ is mathematically expressed by the gamble $\max(x - k, 0)$.}

In what follows we verify under which circumstances the  dual is still the classical theory of probability. Then, later on we will see that sometimes we can get a weaker theory of probabilities. 

Let  $\gambles_R$ denote a closed linear subspace of $\gambles$ that includes all constant gambles. By  $\nonnegative_R = \nonnegative \cap \gambles_R$ we denote the subset of nonnegative gambles, and by  $\negative_R = \negative \cap \gambles_R$ we denote the subset of negative gambles. 
We thus relativise postulates \ref{eq:taut}--\ref{eq:sl} as follows. 
Firstly, we restrict tautologies, and thus the extension of the predicate being nonnegative, to the considered subspace.
\begin{enumerate}[label=\upshape A$_R$0.,ref=\upshape A$_R$0]
\item\label{eq:tautR} $\nonnegative_R$ should always be desirable.
\end{enumerate}

Secondly, we modify the deductive closure accordingly.
\begin{enumerate}[label=\upshape A$_R$1.,ref=\upshape A$_R$1]
\item\label{eq:NER} $N_R(\assess_R)\coloneqq\cl \posi(\nonnegative_R\cup \mathcal{G}_R)$.
\end{enumerate}
Notice that $N_R(\assess_R)\subseteq \gambles_R$, for $\assess_R \subseteq \gambles_R$. In such case we sometimes denote $N_R(\assess_R)$ by $\domain_R$. 

Hence, finally, we state that: 
\begin{definition}[Relative coherence postulate]
\label{def:avsR}
 A set $\domain_R \subseteq \gambles_R$  is \emph{coherent relative to } $\gambles_R$  if and only if
 \begin{enumerate}[label=\upshape A$_R$2.,ref=\upshape A$_R$2]
\item\label{eq:slR} $ \negative_R \cap \domain_R=\emptyset$.
\end{enumerate}
\end{definition}

It is easy to check that $\theory$ and its relativisation to $\gambles_R$ defined by Postulates \ref{eq:tautR}--\ref{eq:slR}, denoted by $\theory_R$, are fully compatible. That is: 

\begin{theorem}\label{thm:cocore}
Let $\mathcal{G}_R$ be a set of assessments in $\gambles_R$.
$N_R(\mathcal{G}_R)$ is coherent in $\theory_R$ if and only if $N(\mathcal{G}_R)$ is coherent in $\theory$. Moreover $N_R(\mathcal{G}_R)= \gambles_R \cap N(\mathcal{G}_R)$.
\end{theorem}

 
 Finally, we compare inference in theory $\theory_R$ with inference in the classical theory $\theory$.
 
 Let $\assess$ be a finite set of assessment in $\gambles_R$, and $N_R(\assess)$ be a coherent set of
desirable gambles in $\theory_R$. The lower prevision of a gamble $f \in \gambles_R$  is defined  as
\begin{equation}
\label{eq:lpner}
\begin{aligned}
 \underline{E}_R(f):=&\sup_{\gamma_0\in \reals, \lambda_i \in \reals^+} \gamma_0\\
 &s.t:\\
 &f - \gamma_0\mathbbm{1}_R  -\sum\limits_{i=1}^{|\mathcal{G}|} \lambda_i g_i \in \nonnegative_R,
\end{aligned}
\end{equation}
where $\mathbbm{1}_R$ denotes the unitary gamble in $\gambles_R$, i.e.,  $\mathbbm{1}_R(\omega)=1$ for all $\omega \in  \Omega$.
The upper prevision is denoted as $ \overline{E}_R(f)=-\underline{E}_R(-f)$. Comparing \eqref{eq:lpne0} and \eqref{eq:lpner}, the reader can notice that
 $f - \gamma_0\mathbbm{1}  -\sum_{i=1}^{|\mathcal{G}|} \lambda_i g_i \in \nonnegative$ in  \eqref{eq:lpne0}
becomes $f - \gamma_0\mathbbm{1}_R  -\sum_{i=1}^{|\mathcal{G}|} \lambda_i g_i \in \nonnegative_R$ in \eqref{eq:lpner}.

Since the inclusion in $\nonnegative$ is a weaker constraint then the inclusion in $\nonnegative_R$,
we have that $\underline{E}_R(f)\leq \underline{E}(f)$.
However, since $\gambles_R$ is a linear subspace of $\gambles$ that includes all constant gambles, with $\nonnegative_R = \nonnegative \cap \gambles_R$, and $N_R(\assess)$ is coherent (in $\theory$), we therefore have that, for every $f \in \gambles_R$
$$
\stackrel{\substack{\theory_R\\~}}{\underline{E}_R(f)} ~~=~~ \stackrel{\substack{\theory\\~}}{\underline{E}(f)} ~~\leq ~~  \stackrel{\substack{\theory\\~}}{\overline{E}(f)}  ~~=~~ \stackrel{\substack{\theory_R\\~}}{\overline{E}_R(f)},
$$
meaning that $\underline{E}_R(f)$ can be interpreted as a lower expectation. In particular, this tell us that, from an operational point of view, $\theory_R$ works exactly as $\theory$.

The same can also be observed by looking at characteristics of the dual of a coherent set in $\theory_R$.
Indeed,   assume $\gambles_R$ is a closed subspace of $\gambles$. Then, endowed with the relative topology, the continuous linear functions on $\gambles_R$ are exactly the restriction to $\gambles_R$ of the continuous linear functionals on $\gambles$  \citep[Theorem 5.87]{aliprantisborder}. 
As $\gambles_R$ includes all constant gambles, by the  Riesz-Kantorovich extension theorem, any positive functional on $\gambles_R$ can be extended (possibly non uniquely) to a positive functional on $\gambles$. This means that, given the correspondence in Theorem \ref{thm:cocore},  the dual of a set $\domain_R$ coherent in $\theory_R$ can be identified with the closed convex set of probability charges $(N(\domain_R))^\bullet$.

To conclude, whenever the subspace $\gambles_R$ is clear and Theorem \ref{thm:cocore} holds, we can identify $\theory$ and $\theory_R$.

 \section{Technicalities for Section \ref{sec:complex}}\label{app:complex}
 
 Define
 \begin{equation}
 f \leq_B g  :=  g-f \in \bnonnegative
 \end{equation}
  \begin{lemma}\label{lem:closure}
   Assume that $\bnonnegative$ includes all positive constant gambles. Then 
   it holds that
   \begin{equation}
       f \geq_B 0  \Leftrightarrow  f + \delta \geq_B 0, \forall\delta > 0.
   \end{equation}
  \end{lemma}
  \begin{proof}
   Direction from left to right follows from  $\delta \in \bnonnegative$ and $\posi \bnonnegative = \bnonnegative$. Direction from right to left follows from $\cl \bnonnegative = \bnonnegative$.
  \end{proof}

  As usual, the natural extension operator is defined  as $N_B(\assess):=\cl(\posi (\assess \cup \bnonnegative))$. In what follows we provide some characterisation of P-coherence.

  \begin{proposition}\label{prop:cohe}
    Assume that $\bnonnegative$ includes all positive constant gambles.
Let $\assess \subseteq \gambles_R$ a set of assessments. The following are equivalent
\begin{enumerate}
 \item $-\mathbbm{1}_R \notin \posi ( \assess \cup \bnonnegative)$
    \item $\posi ( \assess \cup \bnonnegative) \cap \bnegative = \emptyset$
     \item $N_B(\assess)$ is P-coherent
\end{enumerate}
  \end{proposition}
  \begin{proof}
We start by proving the equivalence between the last two points. For the remaining equivalences, first of all, notice that $\posi ( \assess \cup \bnonnegative) \subseteq \cl \posi ( \assess \cup \bnonnegative) \subseteq N_B(\assess)$. Hence, (3) implies (2). Assume that (2) holds, and assume $f + \delta \in \posi (\assess \cup \bnonnegative)$, for every $\delta > 0$. This means $f \in \cl(\posi (\assess \cup \bnonnegative))$.  Suppose $f <_B 0$. By Lemma \ref{lem:closure}, we have $f + \delta <_B 0$, for some $\delta > 0$, that is $f + \delta \in \bnegative$, a contradiction. We therefore conclude that $f \not<_B 0$, and that $\cl(\posi (\assess \cup \bnonnegative))$, which includes $\assess$, is coherent in $\btheory$. 
For the remaining equivalences, clearly (1) implies (2). Now, assume $f \in \bnegative$ and $f \in  \posi ( \assess \cup \bnonnegative)$. Hence, $-f$ is in the interior of $\bnonnegative$, meaning that for some $\delta > 0$, $-f - \delta = g \in \bnonnegative$. From this we get that $-\mathbbm{1}_R = \frac{g + f}{\delta} \in  \posi ( \assess \cup \bnonnegative)$.
  \end{proof}
  
  Analogously to $\theory$, one can ask if and when $N_B$ is a closure operator whose class of non-trivial closed sets coincide with the P-coherent sets, or stated otherwise, if and when  $N_B$ associates to each $\bdomain \subseteq \gambles_R$ the intersection of all P-coherent sets that include $\bdomain$. It turns out that we need to add some conditions to the structural properties of $\btheory$ to obtain such property:
  
  \begin{proposition}\label{prop:eqq}
Assume that $\bnonnegative$ includes all positive constant gambles, and  moreover that
\begin{description}
\item[(*)] for every $f \in \gambles_R$, there is $\epsilon > 0$ such that $f+\epsilon \in \bnonnegative$.
\end{description}
Let $\assess \subseteq \gambles_R$ a set of assessments. The following are equivalent
\begin{enumerate}
     \item $N_B(\assess)$ is P-coherent
    \item $N_B(\assess) \neq \gambles_R$
\end{enumerate}
  \end{proposition}
\begin{proof}
Since (1) implies (2), we need to verify the other direction. Assume $N_B(\assess)$ is not P-coherent. By Proposition \ref{prop:cohe}, $-1 \in N_B(\assess)$, and thus $- \epsilon \in N_B(\assess)$, for every $\epsilon \geq 0$. Let $f \in \gambles_R$. By (*) there is $\epsilon > 0$ such that $f+\epsilon \in \bnonnegative \subseteq N_B(\assess)$.  Hence, by closure under linear combinations,  $f + \epsilon + (- \epsilon) = f \in  N_B(\assess)$.
\end{proof}
  Notice that condition (*) is satisfied by the P-coherent model $\btheory$ of QM introduced in Section \ref{sec:coheqm}, as well as by the model of Section \ref{sec:ent_not_only}. On the other hand, (*) is not the only condition to force the equivalence between the two claims in Proposition \ref{prop:eqq}. One could have indeed added the condition that $\leq$ coincides on $\bnonnegative$ with the order $\leq_B$\footnote{Stated otherwise, for every $f, g \in \bnonnegative$, if $f \leq g $ then $f \leq_B g$.} We plan for future work to study the natural condition on $\bnonnegative$ related to the property for $N_B$ to be the closure operator induced by P-coherent sets.

    In what follows, we prove the main theorem of the paper.
  
 \begin{proofof}{Theorems \ref{th:fundamental}}
  Assume that $\bnonnegative$ includes all positive constant gambles and that it is closed (in $\gambles_R$). 
Let $\bdomain \subseteq \gambles_R$ be a P-coherent set of  desirable gambles. We have to verify that the following statements are equivalent:
\begin{enumerate}
   \item $\bdomain$ includes a negative gamble that is not in $\bnegative$.
\item $\posi(\nonnegative\cup \mathcal{G})$ is incoherent, and thus $\mathcal{P}$ is empty
    \item $\bdomain^{\circ}$ is not (the restriction to  $\gambles_R$ of) a closed convex set of mixtures of classical evaluation functionals. 
    \item The extension   $ \bdomain^\bullet$ of $\bdomain^{\circ}$ in the space $\mathcal{M}$ of all charges in $\pspace$ includes only signed charges (quasi-probabilities).
\end{enumerate}
  First of all, notice that the restriction to $\gambles_R$ of the set of all normalised charges that correspond to a bounded linear functionals coincides with $\bdomain^{\bullet}$. Given this, the equivalence between (3) and (4) is immediate, whereas the equivalence between (2) and (4) is given by Theorem \ref{cor:noncoherent}.  We finally verify the equivalence between (1) and (3). In this case, the direction from left to right being obvious, the other direction is due to the fact that $g \leq f$, for every $g \in \bdomain$ and $f \in \posi(\nonnegative\cup \bdomain)\setminus \bdomain$.
  \end{proofof}
  
  The next result provides a necessary and  sufficient condition for the existence of a P-coherent set of  desirable gambles satisfying each claim of  Theorem \ref{th:fundamental}.
    \begin{proposition}\label{prop:funda}
  Assume that $\gambles_R$ includes all positive constant gambles and $\bnonnegative$ is closed (in $\gambles_R$). The following two claims are equivalent
  \begin{itemize}
  \item there is a P-coherent set of  desirable gambles $\bdomain \subseteq \gambles_R$ that includes a negative gamble that is not P-negative
  \item $\bnonnegative \subsetneq \nonnegative_R$
  \end{itemize}
  \end{proposition}
  
   \section{Technicalities for Subsection \ref{subsec:otherax}}\label{app:otherax}
   In this section we discuss how to exploit the correspondence with the system introduced in  \cite{benavoli2016quantum} in the aim of deriving the remaining three axioms of QT.
   
   Recall that a subset $ \He^{n\times n}$ is  said to be coherent if it is a  convex cone $\mathscr{C}$ containing the set  $\{G\in\He^{n\times n}:G\gneq0\}$ of all PSD matrices in $\He^{n\times n}$ and disjoint from the interior of $\{G\in\He^{n\times n}:G \leq 0\}$. We know that there is a bijective correspondence between closed convex sets of density matrices and coherent subsets of of $\He^{n\times n}$, but also between closed convex sets of density matrices and P-coherent sets of gambles. 
By looking at such correspondences, it is then immediate to verify that:
   
   \begin{proposition}
   The map $\mathfrak{f}: \bdomain \mapsto \{ G \in \He^{n\times n} \mid  x^\dagger G x \in \bdomain\}$ is a bijection between  coherent subsets of $\He^{n\times n}$ and P-coherent sets of gambles.
   \end{proposition}
   Based on this correspondence, we can thus exploit the results in  \cite{benavoli2016quantum} to derive  L\"{u}der's rule (measurement updating) and Schr\"odinger's rule (time evolution).
   
  \subsection{L\"uder's rule}
It states the following:
   
\begin{itemize}
\item
Quantum projection measurements are described by a collection  
$\{\Pi_i\}_{i=1}^n$  of
projection operators that satisfy the completeness equation $\sum_{i=1}^n \Pi_i =I$. These are operators acting on the 
state space of the
system being measured. If the state of the quantum system is $\rho$ immediately
before the measurement then  the state after the measurement is
$$
\hat{\rho}=\dfrac{\Pi_i \rho \Pi_i}{Tr(\Pi_i \rho \Pi_i)},
$$
provided that $Tr(\Pi_i \rho \Pi_i)>0$ and the probability that result $i$ occurs is given by
$p_i=Tr(\Pi_i \rho \Pi_i)$.
\end{itemize}

A projection-valued measurement $\Pi^{*}=\{\Pi_i^*\}_{i=1}^n$ can be seen as a partition of unity $\{x^\dagger\Pi_i^*x \}_{i=1}^n$. Thus an event ``indicated'' by a certain projector $\Pi_i$
in $\Pi=\{\Pi_i\}_{i=1}^n$ can also be seen as the function $\pi_i(x):=x^\dagger\Pi_i^*x$. The information it represents is: an experiment $\Pi$ is performed and the event indicated by $\Pi_i$ happens.\footnote{
We assume  that the quantum measurement device is a ``perfect meter'' (an ideal common  assumption in QM), i.e., there are not observational errors -- Alice can trust the
received information.} Under this assumption, Alice can focus on gambles that are contingent on the event $\Pi_i$: these are the gambles such that ``outside'' $\Pi_i$ no utile is received or due -- status quo is maintained --; in other words, they represent gambles that are called off if the outcome of the experiment is not $\Pi_i$. 
Mathematically, we define Alice's conditional set of desirable gambles as follows.

\begin{definition}
 Let  $\bdomain$ be an P-coherent set of gambles, the set obtained as
\begin{equation}
\label{eq:condition}
\bdomain_{\Pi_i}=\left\{g(x) \in \gambles_R \mid  
x^\dagger \Pi_i G \Pi_i  x\in \bdomain \right\}
\end{equation} 
is  called the {\bf set of desirable gambles conditional} on  $\Pi_i$. \end{definition}

Notice that $x^\dagger \Pi_i G \Pi_i  x= \alpha \pi_i(x)$.

We can also compute the dual of  $\bdomain_{\Pi_i}$, i.e.,
$\mathcal{Q}_{\Pi_i}$, and thus obtain the
\begin{description}
\item[Subjective formulation of L\"uder's rule:]~\\
Given a closed convex set of states $\mathcal{Q}$, the corresponding conditional set on $\Pi_i$ is obtained as
\begin{equation}
\label{eq:rhobayes}
 \mathcal{Q}_{\Pi_i}=\left\{ \dfrac{\Pi_i Z \Pi_i}{Tr(\Pi_i Z \Pi_i)} \Big|  z \in 
\mathcal{Q}\right\},
\end{equation} 
 provided that $Tr(\Pi_i Z \Pi_i)>0$  for every $z \in \mathcal{Q}$. Note that the latter  condition  implies that $\pi_j(x) \notin \mathcal{Q}$
 for any $j\neq i$.
\end{description} 

The following diagram gives the relationships among  $\bdomain,\mathcal{Q}, \bdomain_{\Pi_i},\mathcal{Q}_{\Pi_i}$.
$$
\begin{tikzpicture}
  \matrix (m) [matrix of math nodes,row sep=3em,column sep=6em,minimum width=2em]
  {
     \bdomain & \,\bdomain_{\Pi_i} \\
     \mathcal{Q} & \mathcal{Q}_{\Pi_i} \\};
  \path[-stealth]
    (m-1-1) edge [<->] node [left] {dual} (m-2-1)
            edge [double] node [below] {conditioning} (m-1-2)
    (m-2-1.east|-m-2-2) edge [double]  node [below] {conditioning}
          (m-2-2)
    (m-2-2) edge [<->] node [right] {dual} (m-1-2);
\end{tikzpicture}
$$

\subsection{Time evolution postulate}
It states the following:
   
\begin{itemize}\item
The evolution of a closed quantum system is described by a unitary
transformation. That is, the state $\rho$ of the system at time $t_0$ is related 
to the state
$\rho'$ of the system at time $t_1>t_0$ by a unitary operator $U$ which depends only 
on the
times $t_0$ and $t_1$, $\rho' = U \rho U^\dagger$.
\end{itemize}

Let us consider the dynamics of sets of gambles at present time $t_0$ and future time $t_1$ under the assumption that no information at all is received during such an interval of time (i.e., we have a closed quantum system). The focus is on characterising the coherence of sets of gambles in this time. 

To this end, we add the following temporal postulate:
\begin{enumerate}[label=\upshape B3.,ref=\upshape B3]
\item\label{eq:b3}  A temporal P-coherent transformation is a map $\phi(\cdot,t_1,t_0)$ from $\gambles_R$ to itself that satisfies the following properties: 
 \begin{itemize}
 \item[(i)] $\phi(\cdot,t_0,t_0)$ is the identity map; 
 \item[(ii)]  $\phi(\cdot,t_1,t_0)$ is onto; 
 \item[(iii)] $\phi(\gambles^+,t_1,t_0)=\gambles^+$;
\item[(iv)] $\phi(\cdot,t_1,t_0)$ is linear and constant preserving.
\end{itemize}
 \end{enumerate}
The rationale behind these conditions is the following.

Condition~(i) is obvious. 

Condition~(ii) is a way of stating that sets of desirable gambles are only established at present time $t_0$, since
any gamble $g_1$ at time $t_1$ corresponds to an element $g_0$  at time $t_0$. 

Condition~(iii)  means that no further information is received from time  $t_0$ to time  $t_1$.

Finally, condition~(iv) states in particular once again that the utility scale is linear.

By \cite{benavoli2016quantum}[Theorem A.9] and the fact that $\mathfrak{f}$ is a bijection preserving coherence, the temporal P-coherence postulate \eqref{eq:b3} leads to the following:
   
\begin{description}
\item[Subjective formulation of the time evolution postulate of QT:]~\\
(1) All the transformations $\phi(\cdot,t_1,t_0)$ defined above
are of the following form
$$
g(x)  \xhookrightarrow{\phi} h(x)=x^\dagger (U^\dagger G U) x ,
$$
for some unitary or anti-unitary matrix $U \in \He^{n\times n}$, which only depends on the times $t_1,t_0$ and is equal to the identity
for $t_1=t_0$.\\
(2)  The transformation $\phi(\cdot,t_1,t_0)$  preserves P-coherence:\vspace{0.3cm}\\
\centerline{if $\bdomain_0 $ is P-coherent, then $\bdomain_1=\{g(x) \in \gambles_R \mid x^\dagger (U^\dagger G U) x \in \bdomain_0\}$ is  also P-coherent.}
\end{description}

By exploiting duality, we can also reformulate the above results  in terms of sets of states and derive the time evolution postulate as a direct consequence of temporal P-coherence.
  
\end{appendix}

\bibliographystyle{ieeetr}
\bibliography{biblio}

\end{document}